\documentclass[a4paper]{article}
\usepackage{fullpage}
\usepackage{amsthm}
\usepackage{amsmath}
\usepackage{amssymb,authblk}
\usepackage{t1enc}
\usepackage[utf8]{inputenc}
\usepackage{float}
\usepackage{tikz}
\usepackage{cite}
\usepackage{amsrefs}
\usetikzlibrary{calc}
\usetikzlibrary{decorations}
\usetikzlibrary{decorations.pathreplacing}

\newtheorem{thm}{Theorem}[section]
\newtheorem{lemma}[thm]{Lemma}
\newtheorem{prop}[thm]{Proposition}
\newtheorem{corollary}[thm]{Corollary}
\newtheorem{defi}[thm]{Definition}
\newtheorem{conj}[thm]{Conjecture}
\newtheorem{claim}[thm]{Claim}
\newtheorem{notation}[thm]{Notation}

\begin{document}
\title{The complexity of recognizing minimally tough graphs}

\author[1,2]{Gyula Y Katona\thanks{kiskat@cs.bme.hu}}
\author[3]{Istv\'an Kov\'acs\thanks{istvan.kovacs@shapr3d.com}}
\author[1,4]{Kitti Varga\thanks{vkitti@cs.bme.hu}}
\affil[1]{Department of Computer Science and Information Theory, Budapest University of Technology and Economics, Hungary}
\affil[2]{MTA-ELTE Numerical Analysis and Large Networks Research Group, Hungary}
\affil[3]{Shapr3D, Hungary}
\affil[4]{HAS Alfr\'ed R\'enyi Institute of Mathematics, Hungary}
\maketitle

\begin{abstract}
 A graph is called \mbox{$t$-tough} if the removal of any vertex set~$S$ that disconnects the graph leaves at most $|S|/t$ components. The toughness of a graph is the largest~$t$ for which the graph is \mbox{$t$-tough}. A graph is minimally \mbox{$t$-tough} if the toughness of the graph is~$t$ and the deletion of any edge from the graph decreases the toughness. The complexity class DP is the set of all languages that can be expressed as the intersection of a language in NP and a language in coNP. In this paper, we prove that recognizing minimally \mbox{$t$-tough} graphs is DP-complete for any positive rational number~$t$. We introduce a new notion called weighted toughness, which has a key role in our proof.
\end{abstract}

\section{Introduction}

All graphs considered in this paper are finite, simple and undirected. Let $\omega(G)$ denote the number of components and $\alpha(G)$ denote the independence number of a graph~$G$. (Using $\omega(G)$ to denote the number of components may be confusing, however, most of the literature on toughness uses this notation.) For a graph~$G$ and a vertex set $V' \subseteq V(G)$, let $G[V']$ denote the subgraph of~$G$ induced by~$V'$. For a connected graph $G$, a vertex set $S \subseteq V(G)$ is called a cutset if its removal disconnects the graph.

The notion of toughness was introduced by Chv\'atal~\cite{toughness_intro}.

\begin{defi}
 Let $t$ be a real number. A graph~$G$ is called {\em \mbox{$t$-tough}} if $|S| \ge t \omega(G-S)$ holds for any vertex set $S \subseteq V(G)$ that disconnects the graph (i.e. for any $S \subseteq V(G)$ with $\omega(G-S)>1$). The \emph{toughness} of~$G$, denoted by $\tau(G)$, is the largest~$t$ for which~$G$ is \mbox{$t$-tough}, taking $\tau(K_n) = \infty$ for all $n \ge 1$.
 
 We say that a cutset $S \subseteq V(G)$ is a {\em tough set} if $\omega(G - S) = |S|/\tau(G)$.
\end{defi}

It follows directly from the definition of toughness that every \mbox{$t$-tough} noncomplete graph is \mbox{$2t$-}con\-nected; therefore, the minimum degree of any \mbox{$t$-tough} noncomplete graph is at least~$\lceil 2t \rceil$.

Clearly, the more edges a graph has, the larger its connectivity can be, so the graphs whose toughness decreases whenever one of their edges are removed might have some interesting properties.

\begin{defi}
 A graph~$G$ is \emph{minimally \mbox{$t$-tough}} if $\tau(G) = t$ and $\tau(G - e) < t$ for all $e \in E(G)$.
\end{defi}

The motivation for our research is the following conjecture.

\begin{conj}[Kriesell \cite{kriesell}] \label{kriesell}
 Every minimally $1$-tough graph has a vertex of degree~$2$.
\end{conj}

The above conjecture can be naturally generalized to any positive rational number~$t$ as follows: every minimally \mbox{$t$-tough} graph has a vertex of degree~$\lceil 2t \rceil$. Note that this conjecture is an analogue of a theorem of Mader stating that for any positive integer~$k$, every minimally \mbox{$k$-connected} graph has a vertex of degree~$k$, see~\cite{ende}.

Kriesell's conjecture is still open, but in~\cite{min1tough_article} we presented some related results, in particular that in the class of \mbox{claw-free} graphs the conjecture is true in a very strong sense, namely, the only minimally 1-tough, \mbox{claw-free} graphs are cycles of length at least four. On the other hand, we also proved that the class of minimally \mbox{$t$-tough} graphs is large for any positive rational number~$t$: any graph can be embedded as an induced subgraph into a minimally \mbox{$t$-tough} graph. 

Therefore it is natural to ask, how ``large'' the set of minimally \mbox{$t$-tough} graphs is for different $t$ values and for various graph classes. In the present paper we investigate the first question from a complexity theoretical viewpoint. Similar results for the second question are presented in~\cite{spec_graphclasses}.

Let $t$ be an arbitrary positive rational number and consider the following problem.

\medskip
\noindent {\bf \scshape $\boldsymbol{t}$-Tough} \\
\textit{Instance:} a graph~$G$. \\
\textit{Question:} is it true that $\tau(G) \ge t$?
\medskip

Note that in this problem, $t$ is not part of the input.

It is easy to see that for any positive rational number~$t$, the problem {\scshape \mbox{$t$-Tough}} is in coNP: a witness is a vertex set S whose removal disconnects the graph and leaves more than $|S|/t$ components. By reducing a variant of the independent set problem to the complement of {\scshape \mbox{$t$-Tough}}, Bauer et al.~proved the following.

\begin{thm}[\cite{recognize_toughness}] \label{t_tough_conp_complete}
 For any positive rational number~$t$, the problem {\scshape \mbox{$t$-Tough}} is coNP-complete.
\end{thm}

However, in some graph classes the toughness can be computed in polynomial time, for instance, in the class of split graphs.

\begin{thm}[\cite{split_general}] \label{split_general_thm}
 For any rational number $t > 0$, the class of \mbox{$t$-tough} split graphs can be recognized in polynomial time.
\end{thm}

The focus of our investigation is on the critical version of the problem {\scshape \mbox{$t$-Tough}}. Let $t$ be an arbitrary positive rational number and consider the following problem.

\medskip
\noindent {\bf \scshape Min-$\boldsymbol{t}$-Tough} \\
\textit{Instance:} a graph~$G$.\\
\textit{Question:} is it true that $G$ is minimally \mbox{$t$-tough}?
\medskip

Since extremal problems usually seem not to belong to $\text{NP} \cup \text{coNP}$, the complexity class called DP was introduced by Papadimitriou and Yannakakis in~\cite{dp_intro}.

\begin{defi}
 A language~$L$ is in the class \emph{DP} if there exist two languages $L_1 \in \text{NP}$ and $L_2 \in \text{coNP}$ such that $L = L_1 \cap L_2$.
 
 A language is called \emph{DP-hard} if all problems in DP can be reduced to it in polynomial time. A language is \emph{DP-complete} if it is in DP and it is DP-hard.
\end{defi}

It should be emphasized that $\text{DP} \ne \text{NP} \cap \text{coNP}$ if $\text{NP} \ne \text{coNP}$. Moreover, $\text{NP} \cup \text{coNP} \subseteq \text{DP}$. 

To prove that {\scshape Min-\mbox{$t$-Tough}} is DP-complete for any positive rational number~$t$, we use the following problem for reduction.

\medskip
\noindent {\bf \scshape $\boldsymbol{\alpha}$-Critical} \\
\textit{Instance:} a graph~$G$ and a positive integer~$k$. \\
\textit{Question:} is it true that $\alpha(G) < k$, but $\alpha(G-e) \ge k$ for any edge $e \in E(G)$?
\medskip

Note that, unlike in {\scshape \mbox{$t$-tough}} or {\scshape Min-\mbox{$t$-Tough}}, in this problem $k$ is part of the input. The DP-completeness of the problem {\scshape $\alpha$-Critical} is a trivial consequence of the following theorem.

\begin{thm}[\cite{crit_clique}]
 The following problem is DP-complete.
 
 \medskip
 \noindent {\bf \scshape CriticalClique} \\
 \textit{Instance:} a graph~$G$ and a positive integer~$k$. \\
 \textit{Question:} is it true that $G$ has no clique of size~$k$, but adding any missing edge~$e$ to~$G$, the resulting graph $G+e$ has a clique of size~$k$?
\end{thm}

\begin{corollary} \label{alpha_crit_dp_complete}
 The problem {\scshape $\alpha$-Critical} is DP-complete.
\end{corollary}

Our main result is the following.

\begin{thm} \label{minttough_dp_complete}
 The problem {\scshape Min-\mbox{$t$-Tough}} is DP-complete for any positive rational number~$t$.
\end{thm}

Note that since the toughness of any noncomplete graph is a rational number, there exist no minimally tough graphs with irrational toughness.

To prove the case $t \ge 1$, we introduce a new notion called weighted toughness. However, we believe that this might be an interesting idea on its own.

\begin{defi}
 Let $t$ be a positve real number. Given a graph~$G$ and a positive weight function~$w$ on its vertices, we say that the graph~$G$ is weighted \mbox{$t$-tough} with respect to the weight function~$w$ if
  \[ \omega(G-S) \le \frac{w(S)}{t} \]
 holds for any vertex set $S \subseteq V(G)$ whose removal disconnects the graph, where 
  \[ w(S) = \sum_{v \in S} w(v) \text{.} \]
 The weighted toughness of a noncomplete graph (with respect to the weight function~$w$) is the largest~$t$ for which the graph is weighted \mbox{$t$-tough}, and we define the weighted toughness of complete graphs (with respect to~$w$) to be infinity. 
\end{defi}

Note that the weighted toughness of a graph with respect to the weight function that assigns~1 to every vertex is the toughness of the graph.

The paper is organized as follows. Section~\ref{preliminaries} gathers the properties of minimally tough graphs and \mbox{$\alpha$-critical} graphs that are needed to prove Theorem~\ref{minttough_dp_complete}. Since the proof of this theorem is fairly complicated, Section~\ref{section_special_cases} discusses some of its special cases in the hope of fostering a better understanding. The proof of Theorem~\ref{minttough_dp_complete} considers three cases: when $1/2 < t < 1$, when $t \ge 1$, and when $t \le 1/2$; they are proved in Sections~\ref{section_1/2<t<1}, \ref{section_t>=1} and~\ref{section_t<=1/2}, respectively.

\section{Preliminaries} \label{preliminaries}

In this section we collect some useful properties of minimally tough graphs and \mbox{$\alpha$-critical} graphs.

\subsection{Minimally tough graphs}

\begin{prop} \label{obs_below1}
 Let $t \le 1$ be a positive rational number and $G$ a graph with $\tau(G) = t$. Then 
  \[ \omega(G-S) \le \frac{|S|}{t} \]
 for any nonempty proper subset~$S$ of~$V(G)$.
\end{prop}

\begin{proof}
 If $S$ is a cutset in~$G$, then by the definition of toughness $\omega(G-S) \le |S|/t$ holds.

 If $S$ is not a cutset in~$G$, then $\omega(G-S) = 1$ (since $S \ne V(G)$). On the other hand, $|S|/t \ge 1$ since $S \ne \emptyset$ and $t \le 1$. Therefore $\omega(G-S) \le |S|/t$ holds in this case as well.
\end{proof}

As is clear from its proof, the above proposition holds even if $S$ is not a cutset. However, it does not hold if $t > 1$ and $S$ is not a cutset: if $t>1$, then the graph cannot contain a cut-vertex; therefore $\omega(G-S) = 1$ for any subset~$S$ with $|S| = 1$, while $|S|/t = 1/t < 1$.

\begin{prop} \label{possible_values_of_toughness}
 Let $G$ be a connected noncomplete graph on $n$~vertices. Then $\tau(G)$ is a positive rational number, and if $\tau(G) = a/b$, where $a,b$ are relatively prime positive integers, then $1 \le a,b \le n-1$.
\end{prop}
\begin{proof}
 By definition, 
  \[ \tau(G) = \min_{\substack{S \subseteq V(G) \\ \omega(G-S) \ge 2}} \frac{|S|}{\omega(G-S)} \]
 for a noncomplete graph~$G$. Since $G$ is connected and noncomplete, $1 \le |S| \le n-2$ for every $S \subseteq V(G)$ with $\omega(G-S) \ge 2$. Obviously, $\omega(G-S) \ge 2$, and since $G$ is connected, $\omega(G-S) \le n-1$.
\end{proof}

\begin{corollary} \label{toughness_gap}
 Let $G$ and $H$ be two connected noncomplete graphs on $n$~vertices. If $\tau(G) \ne \tau(H)$, then
  \[ \big| \tau(G) - \tau(H) \big| > \frac{1}{n^2} \text{.} \]
\end{corollary}
\begin{proof}
 Let $a,b$ and $a',b'$ be two pairs of relative prime positive integers such that $\tau(G) = a/b$ and $\tau(H) = a' / b'$. Proposition~\ref{possible_values_of_toughness} implies that $1 \le a,b,a',b' \le n-1$. Since $\tau(G) \ne \tau(H)$,
  \[ \big| \tau(G) - \tau(H) \big| = \left| \frac{a}{b} - \frac{a'}{b'} \right| = \left| \frac{ab' - a'b}{bb'} \right| > \frac{1}{n^2} \text{.} \]
\end{proof}

\begin{prop} \label{min_t_tough_DP}
 For every positive rational number~$t$, the problem {\scshape Min-\mbox{$t$-Tough}} belongs to DP.
\end{prop}
\begin{proof}
 For any positive rational number~$t$,
 \begin{gather*}
  \text{{\scshape Min-\mbox{$t$-Tough}}} = \big\{ G \text{ graph} \bigm| \tau(G) = t \text{ and } \tau(G-e) < t \text{ for all } e \in E(G) \big\} \\
  = \big\{ G \text{ graph} \bigm| \tau(G) \ge t \big\} \cap \big\{ G \text{ graph} \bigm| \tau(G) \le t \big\} \\
  \cap \, \big\{ G \text{ graph} \bigm| \tau(G-e) < t \text{ for all } e \in E(G) \big\} \text{.}
 \end{gather*}
 Let
  \[ L_{1,1} = \big\{ G \text{ graph} \bigm| \tau(G-e) < t \text{ for all } e \in E(G) \big\} \text{,} \]
  \[ L_{1,2} = \big\{ G \text{ graph} \bigm| \tau(G) \le t \big\} \]
 and
  \[ L_2 = \big\{ G \text{ graph} \bigm| \tau(G) \ge t \big\} \text{.} \]
 Notice that $L_2 = \text{{\scshape \mbox{$t$-Tough}}}$ and it is known to be in \text{coNP}: if a graph $G$ is not $t$-tough, then a witness is a vertex set $S \subseteq V(G)$ whose removal disconnects~$G$ and leaves more than $|S|/t$ components. Similarly, $L_{1,1} \in \text{NP}$, since a witness is a set of vertex sets $\big\{ S_e \subseteq V(G) \bigm| e \in E(G) \big\}$, where for any $e \in E(G)$ the removal of $S_e$ disconnects $G-e$ and leaves more than $|S_e|/t$ components.

 Now we show that $L_{1,2} \in \text{NP}$, i.e. we can express $L_{1,2}$ in the form
  \[L_{1,2} = \big\{ G \text{ graph} \bigm| \tau(G) < t + \varepsilon \big\} \text{,} \]
 which is the complement of a language belonging to coNP. Let $G$ be an arbitrary graph on $n$~vertices. If $G$ is disconnected, then $\tau(G) = 0$, and if $G$ is complete, then $\tau(G) = \infty$, so in both cases $\tau(G) \le t$ if and only if $\tau(G) < t + \varepsilon$ for any positive number~$\varepsilon$. If $G$ is connected and noncomplete, then from Corollary~\ref{toughness_gap} it follows that $\tau(G) \le t$ if and only if $\tau(G) < t + 1/n^2$. Therefore
  \[ L_{1,2} = \big\{ G \text{ graph} \bigm| \tau(G) \le t \big\} = \left\{ G \text{ graph} ~ \middle| ~ \tau(G) < t + \frac{1}{|V(G)|^2} \right\} \text{,} \]
 so $L_{1,2} \in \text{NP}$.
 
 Since $L_{1,1} \cap L_{1,2} \in \text{NP}$ and $L_2 \in \text{coNP}$ and $\text{{\scshape Min-\mbox{$t$-Tough}}} = \left( L_{1,1} \cap L_{1,2} \right) \cap L_2$, we can conclude that $\text{{\scshape Min-\mbox{$t$-Tough}}} \in \text{DP}$.
\end{proof}

\begin{prop} \label{minttoughlemma}
 Let $t$ be a positive rational number and $G$ a minimally \mbox{$t$-tough} graph. For every edge~$e$ of~$G$,
 \begin{enumerate}
  \item the edge~$e$ is a bridge in~$G$, or
  \item there exists a vertex set ${S=S(e) \subseteq V(G)}$ with 
    \[ \omega(G-S) \le \frac{|S|}{t} \quad \text{and} \quad \omega \big( (G-e)-S \big) > \frac{|S|}{t} \text{,} \]
   and the edge~$e$ is a bridge in $G-S$.
 \end{enumerate}
 In the first case, we define $S = S(e) = \emptyset$.
\end{prop}
\begin{proof}
 Let $e$ be an arbitrary edge of~$G$ which is not a bridge. Since $G$ is minimally \mbox{$t$-tough}, ${\tau (G-e) < t}$. Since $e$ is not a bridge, $G-e$ is still connected, so there exists a cutset $S=S(e) \subseteq V(G-e)=V(G)$ in $G-e$ satisfying $\omega \big( (G-e)-S \big) > |S|/t$.
 
 By Proposition~\ref{obs_below1}, if $t \le 1$, then $\omega (G-S) \le |S|/t$. So assume that $t > 1$. Now there are two cases.
 
 \bigskip
 
 \textit{Case 1:} ($t > 1$ and) $S$ is a cutset in~$G$.
 
 Since $\tau (G) = t$ and $S$ is a cutset, $\omega (G-S) \le |S|/t$. This is only possible if $e$ connects two components of $(G-e)-S$.
 
 \bigskip
 
 \textit{Case 2:} ($t > 1$ and) $S$ is not a cutset in~$G$. 
 
 Then $\omega(G-S) = 1$. Since $S$ is a cutset in $G-e$, the edge~$e$ must connect two components of $(G-e)-S$, so 
 \[ \omega \big( (G-e) - S \big) = 2 \text{.} \]
 Now we show that $\omega(G-S) \le |S|/t$. Suppose to the contrary that $\omega(G-S) > |S|/t$. Since $\omega(G-S) = 1$, this implies that $|S| < t$. Moreover, since $\tau(G) = t$, the graph~$G$ is \mbox{$\lceil 2t \rceil$-}con\-nected, and thus it has at least $2t+1$ vertices. From this it follows that $S$ and one of the endpoints of $e$ form a cutset in $G$, otherwise $G$ would only have 
 \[ |S|+2 < t+2 < 2t+1 \]
 vertices (where the latter inequality is valid since $t > 1$). Let $S'$ denote this cutset. Since $G$ is \mbox{$t$-tough} and $S'$ is a cutset in~$G$,
  \[ 2 \le \omega(G-S') \le \frac{|S'|}{t} = \frac{|S|+1}{t} \text{,} \]
 so $|S| \ge 2t-1$. Therefore
  \[ 2t-1 \le |S| < t \text{,} \]
 which implies that $t < 1$ and that is a contradiction.
\end{proof}

\subsection{Almost minimally 1-tough graphs}

The graphs~$K_2$ and~$K_3$ behave similarly as minimally \mbox{1-tough} graphs: they are \mbox{1-tough}, and the removal of any of their edges decreases their toughness below 1. However, they are not minimally \mbox{1-tough} since their toughness is infinity. To handle these kinds of graphs, we introduce the following definition.

\begin{defi}
 A graph~$G$ is \emph{almost minimally \mbox{$1$-tough}} if $\tau(G) \ge 1$ and $\tau(G - e) < 1$ for all $e \in E(G)$.
\end{defi}

In fact, the only almost minimally \mbox{1-tough} graphs are minimally \mbox{1-tough} graphs and the graphs~$K_2$ and~$K_3$.

\begin{claim} \label{almostmin1tough_equivalent_forms}
 For a graph~$G$ the following are equivalent.
 \begin{enumerate}
  \item[$(1)$] The graph~$G$ is almost minimally \mbox{1-tough}.
  \item[$(2)$] The graph~$G$ is \mbox{1-tough} and for every $e \in E(G)$, the edge~$e$ is a bridge or there exists a vertex set $S = S(e) \subseteq V(G)$ with
   \[ \omega(G-S) = |S| \qquad \text{and} \qquad \omega \big( (G-e)-S \big) = |S| + 1 \text{.} \]
  (If $e$ is a bridge, we define $S=S(e)=\emptyset$.)
  \item[$(3)$] The graph~$G$ is either minimally \mbox{1-tough} or $G \simeq K_2$ or $G \simeq K_3$.
 \end{enumerate}
\end{claim}
\begin{proof} $\empty$

 $(1) \Longrightarrow (2):$ Let $e$ be an arbitrary edge of~$G$, and let us assume that it is not a bridge. Since $\tau (G-e) < 1$ and $G-e$ is still connected, there exists a cutset $S=S(e) \subseteq V(G-e)=V(G)$ in $G-e$ satisfying $\omega \big( (G-e)-S \big) > |S|$.
 
 Now there are two cases.
 
 \bigskip
 
 \textit{Case 1:} $S$ is a cutset in~$G$.
 
 Since $\tau (G) \ge 1$ and $S$ is a cutset, $\omega (G-S) \le |S|$. This is only possible if $e$ connects two components of $(G-e)-S$, which means that 
  \[ \omega \big( (G-e)-S \big) = |S| + 1 \qquad \text{and} \qquad \omega (G-S) = |S| \text{.} \]
 
 \bigskip
 
 \textit{Case 2:} $S$ is not a cutset in~$G$.
 
 Then $\omega(G-S) = 1$. On the other hand,
  \[ \omega \big( (G-e)-S \big) \ge 2 \]
 since $S$ is a cutset in $G-e$. This is only possible if $e$ connects two components of $(G-e)-S$, which means that
 \[ \omega \big( (G-e)-S \big) = 2 \text{.} \]
 Since
  \[ \omega \big( (G-e)-S \big) > |S| \text{,} \]
 this implies that $|S| \le 1$. Moreover, $|S| = 1$ since $e$ is not a bridge in~$G$. Hence,
  \[ \omega \big( (G-e)-S \big) = |S| + 1 \qquad \text{and} \qquad \omega (G-S) = |S| \text{.} \]
 
 \bigskip
 
 $(2) \Longrightarrow (3):$ Then $\tau(G) \ge 1$ and $\tau(G-e) < 1$ for every $e \in E(G)$. Let us assume that $G$ is not minimally \mbox{1-tough}, i.e.~$\tau(G) > 1$. We need to show that $G \simeq K_2$ or $G \simeq K_3$.
 
 Suppose to the contrary that $G$ has at least 4~vertices. Let $e \in E(G)$ be an arbitrary edge, and let $S = S(e) \subseteq V(G)$ be a vertex set for which 
  \[ \omega(G-S) = |S| \qquad \text{and} \qquad \omega \big( (G-e)-S \big) = |S| + 1 \text{.} \]
 Since $\tau(G) > 1$ and $\omega(G-S) = |S|$, the vertex set~$S$ cannot be a cutset in~$G$, so $|S| \le 1$ must hold. Since $G$ has at least 4~vertices, $S$ and one of the endpoints of~$e$ form a cutset of size at most~2, so $\tau(G) \le 1$, which is a contradiction. This means that $G \simeq K_2$ or $G \simeq K_3$, since there are no other \mbox{1-tough} graphs on at most~3 vertices with at least one edge.
 
 \bigskip
 
 $(3) \Longrightarrow (1):$ Trivial.
\end{proof}

\begin{prop} \label{obs_almostmin1tough}
 Let $G$ be an almost minimally \mbox{1-tough} graph. Then $\omega(G-S) \le |S|$ for any nonempty proper subset~$S$ of~$V(G)$.
\end{prop}
\begin{proof}
 By Claim~\ref{almostmin1tough_equivalent_forms}, the graph~$G$ is either minimally \mbox{1-tough} or $G \simeq K_2$ or $G \simeq K_3$. If $G$ is minimally \mbox{1-tough}, then $\tau(G) = 1$, and we already covered this case in Proposition~\ref{obs_below1}. If $G \simeq K_2$ or $G \simeq K_3$, then $\omega(G-S) = 1$ and $1 \le |S| \le 2$ hold for any nonempty proper subset~$S$ of~$V(G)$.
\end{proof}

\subsection{$\alpha$-critical graphs}

First, we cite some results on \mbox{$\alpha$-critical} graphs.

\begin{prop}[Problem $12$ of \S 8 in \cite{lovasz}] \label{alphacrit_max_indep}
 If $G$ is an \mbox{$\alpha$-critical} graph without isolated vertices, then every point is contained in at least one maximum independent vertex set.
\end{prop}

\begin{lemma}[Problem $14$ of \S 8 in \cite{lovasz}] \label{alpha_crit_blowup}
 If we replace a vertex of an \mbox{$\alpha$-critical} graph with a clique, and connect every neighbor of the original vertex with every vertex in the clique, then the resulting graph is still \mbox{$\alpha$-critical}.
\end{lemma}

\begin{lemma}[\cite{lovasz_matching}] \label{how_to_increase_alpha}
 Let $G$ be an \mbox{$\alpha$-critical} graph and~$w$ an arbitrary vertex of degree at least~2. Split~$w$ into two vertices~$y$ and~$z$, each of degree at least~1, add a new vertex~$x$ and connect it to both~$y$ and~$z$. Then the resulting graph~$G'$ is \mbox{$\alpha$-critical}, and $\alpha(G') = \alpha(G) + 1$.
\end{lemma}

For one of our proofs we also need the following observation, which is a straightforward consequence of Corollary~\ref{alpha_crit_dp_complete} and Lemmas~\ref{alpha_crit_blowup} and~\ref{how_to_increase_alpha}.

\begin{prop} \label{alphacrit_variant}
 For any positive integers~$l$ and~$m$, the following variant of the problem {\scshape \mbox{$\alpha$-Critical}} is DP-complete.
 
 \medskip
 \normalfont{
 \noindent
 \textit{Instance:} an \mbox{$l$-}con\-nected graph~$G$ and a positive integer~$k$ that is divisible by~$m$. \\
 \textit{Question:} is it true that $\alpha(G) < k$, but $\alpha(G-e) \ge k$ for any edge $e \in E(G)$?}
\end{prop}

\section{On some special cases of Theorem~\ref{minttough_dp_complete}} \label{section_special_cases}

This section aims to highlight the key steps of the proof of Theorem~\ref{minttough_dp_complete} by considering some simpler cases of it. In the view of this intention, technical details are omitted here.

Let $n \ge 2$ be an integer and let $G$ be a complete graph of size~$n$ on the vertex set $\{ v_1, \ldots,v_n \}$. Add the vertices $u_1, \ldots, u_n$ and~$w$ to~$G$, and for all $i \in [n]$ connect $v_i$ and~$u_i$, and also $u_i$ and~$w$, and let $G'$ denote the obtained graph. (For an example see Figure~\ref{Fig:G'} in the Appendix.) It is easy to see that $G'$ is a minimally \mbox{1-tough} graph, and it is due to the fact that complete graphs are \mbox{$\alpha$-critical}. This plain construction inspires all the others proposed in this paper. This construction can be generalized for \mbox{$\alpha$-critical} graphs in general to obtain a minimally 1-tough graph. (See Figure~\ref{Fig:G'12}.) The construction for minimally integer-tough graphs can be seen as a ``blow-up'' of the minimally 1-tough construction. (See Figure~\ref{Fig:G'21}.) These constructions are described in details in the following subsection.

\subsection{On the case of minimally $t$-tough graphs, where $t$ is a positive integer} \label{subsection_integer}

Let $t$, $k$ and $n \ge t+1$ be positive integers, let $G$ be an arbitrary $\big\lceil (t+1)/2 \big\rceil$-connected graph on the vertices $v_1, \ldots, v_n$, and let $G'_{t,k}$ be defined as follows. For all $i \in [n]$ and $j \in [k]$ let
 \[ V_{i,j} = \big\{ v_{i,j,l} \bigm| l \in [t] \big\} \text{.} \]
For all $i \in [n]$ let
 \[ V_i = \bigcup_{j \in [k]} V_{i,j} \]
and place a complete graph on its vertices. For all $i_1,i_2 \in [n]$ if $v_{i_1} v_{i_2} \in E(G)$, then place a complete bipartite graph on $(V_{i_1}; V_{i_2})$. For all $i \in [n]$ and $j \in [k]$ add the vertex set
 \[ U_{i,j} = \big\{ u_{i,j,l} \bigm| l \in [t] \big\} \]
to the graph and place a complete graph on the vertices of $U_{i,j}$. For all $i \in [n], j \in [k], l \in [t]$ connect $v_{i,j,l}$ to $u_{i,j,l}$. For all $j \in [k]$ add the vertex set
 \[ W_j = \{ w_{j,1}, \ldots, w_{j,t} \} \]
to the graph and for all $i \in [n]$ place a complete bipartite graph on $(U_{i,j}; W_j)$.
Let
 \[ V = \bigcup_{i = 1}^n V_i \text{,} \qquad U = \bigcup_{i=1}^n \bigcup_{j=1}^{k} U_{i,j} \text{,} \qquad W = \bigcup_{j = 1}^{k} W_j \text{.} \]
See Figure~\ref{Fig:G'tk_t_integer}. (For examples see Figures~\ref{Fig:G'12} and~\ref{Fig:G'21}.)

\begin{figure}[H]
 \begin{center}
 \begin{tikzpicture}
  \tikzstyle{vertex}=[draw,circle,fill=black,minimum size=3,inner sep=0]
  
  \draw (-0.25,0) ellipse (1 and 3);
  \draw (0,1.45) ellipse (0.3 and 1);
  \draw[fill] (0,0.1) circle (0.3pt);
  \draw[fill] (0,0) circle (0.3pt);
  \draw[fill] (0,-0.1) circle (0.3pt);
  \draw (0,-1.45) ellipse (0.3 and 1);
  
  \node at (-1,3) {$G$};
  \node at (-0.55,2.3) {$V_1$};
  \node at (-0.55,-0.7) {$V_n$};
  
  \node[vertex] (v11) at (0,2.25) {};
  \draw[fill] (0,2.1) circle (0.3pt);
  \draw[fill] (0,2) circle (0.3pt);
  \draw[fill] (0,1.9) circle (0.3pt);
  \node[vertex] (v12) at (0,1.75) {};
  \draw[fill] (0,1.55) circle (0.3pt);
  \draw[fill] (0,1.45) circle (0.3pt);
  \draw[fill] (0,1.35) circle (0.3pt);
  \node[vertex] (v13) at (0,1.15) {};
  \draw[fill] (0,1) circle (0.3pt);
  \draw[fill] (0,0.9) circle (0.3pt);
  \draw[fill] (0,0.8) circle (0.3pt);
  \node[vertex] (v14) at (0,0.65) {};
  
  \node[vertex] (vn1) at (0,-0.65) {};
  \draw[fill] (0,-0.8) circle (0.3pt);
  \draw[fill] (0,-0.9) circle (0.3pt);
  \draw[fill] (0,-1) circle (0.3pt);
  \node[vertex] (vn2) at (0,-1.15) {};
  \draw[fill] (0,-1.35) circle (0.3pt);
  \draw[fill] (0,-1.45) circle (0.3pt);
  \draw[fill] (0,-1.55) circle (0.3pt);
  \node[vertex] (vn3) at (0,-1.75) {};
  \draw[fill] (0,-1.9) circle (0.3pt);
  \draw[fill] (0,-2) circle (0.3pt);
  \draw[fill] (0,-2.1) circle (0.3pt);
  \node[vertex] (vn4) at (0,-2.25) {};
  
  \node[vertex] (u11) at (2,2.25) {};
  \draw[fill] (2,2.1) circle (0.3pt);
  \draw[fill] (2,2) circle (0.3pt);
  \draw[fill] (2,1.9) circle (0.3pt);
  \node[vertex] (u12) at (2,1.75) {};
  \draw[fill] (2,1.55) circle (0.3pt);
  \draw[fill] (2,1.45) circle (0.3pt);
  \draw[fill] (2,1.35) circle (0.3pt);
  \node[vertex] (u13) at (2,1.15) {};
  \draw[fill] (2,1) circle (0.3pt);
  \draw[fill] (2,0.9) circle (0.3pt);
  \draw[fill] (2,0.8) circle (0.3pt);
  \node[vertex] (u14) at (2,0.65) {};
  
  \node[vertex] (un1) at (2,-0.65) {};
  \draw[fill] (2,-0.8) circle (0.3pt);
  \draw[fill] (2,-0.9) circle (0.3pt);
  \draw[fill] (2,-1) circle (0.3pt);
  \node[vertex] (un2) at (2,-1.15) {};
  \draw[fill] (2,-1.35) circle (0.3pt);
  \draw[fill] (2,-1.45) circle (0.3pt);
  \draw[fill] (2,-1.55) circle (0.3pt);
  \node[vertex] (un3) at (2,-1.75) {};
  \draw[fill] (2,-1.9) circle (0.3pt);
  \draw[fill] (2,-2) circle (0.3pt);
  \draw[fill] (2,-2.1) circle (0.3pt);
  \node[vertex] (un4) at (2,-2.25) {};
  
  \draw (2,2) ellipse (0.25 and 0.35);
  \draw (2,0.9) ellipse (0.25 and 0.35);
  \draw (2,-0.9) ellipse (0.25 and 0.35);
  \draw (2,-2) ellipse (0.25 and 0.35);
  
  \node at (1.4,2.5) {$U_{1,1}$};
  \node at (1.4,1.4) {$U_{1,k}$};
  \draw[fill] (2,0.1) circle (0.3pt);
  \draw[fill] (2,0) circle (0.3pt);
  \draw[fill] (2,-0.1) circle (0.3pt);
  \node at (1.4,-0.4) {$U_{n,1}$};
  \node at (1.4,-1.5) {$U_{n,k}$};
  
  \node at (2.4,2.75) {$U$};
  
  \draw (v11) -- (u11);
  \draw (v12) -- (u12);
  \draw (v13) -- (u13);
  \draw (v14) -- (u14);
  \draw (vn1) -- (un1);
  \draw (vn2) -- (un2);
  \draw (vn3) -- (un3);
  \draw (vn4) -- (un4);
  
  \node[vertex] (w11) at (4.25,1) [label=right:{$w_{1,1}$}] {};
  \draw[fill] (4.25,0.85) circle (0.3pt);
  \draw[fill] (4.25,0.75) circle (0.3pt);
  \draw[fill] (4.25,0.65) circle (0.3pt);
  \node[vertex] (w1t) at (4.25,0.5) [label=right:{$w_{1,t}$}] {};
  \draw[fill] (4.25,0.1) circle (0.3pt);
  \draw[fill] (4.25,0) circle (0.3pt);
  \draw[fill] (4.25,-0.1) circle (0.3pt);
  \node[vertex] (wa1) at (4.25,-0.5) [label=right:{$w_{k,1}$}] {};
  \draw[fill] (4.25,-0.65) circle (0.3pt);
  \draw[fill] (4.25,-0.75) circle (0.3pt);
  \draw[fill] (4.25,-0.85) circle (0.3pt);
  \node[vertex] (wat) at (4.25,-1) [label=right:{$w_{k,t}$}] {};
  
  \node at (4.25,1.7) {$W$};
  
  \draw (u11) -- (w11);
  \draw (u12) -- (w1t);
  \draw (u11) -- (w1t);
  \draw (u12) -- (w11);
  
  \draw (un1) -- (w11);
  \draw (un2) -- (w1t);
  \draw (un1) -- (w1t);
  \draw (un2) -- (w11);
  
  \draw (u13) -- (wa1);
  \draw (u14) -- (wat);
  \draw (u13) -- (wat);
  \draw (u14) -- (wa1);
  
  \draw (un3) -- (wa1);
  \draw (un4) -- (wat);
  \draw (un3) -- (wat);
  \draw (un4) -- (wa1);
 \end{tikzpicture}
 \caption{The graph~$G'_{t,k}$, when $t$ is a positive integer.} \label{Fig:G'tk_t_integer}
 \end{center}
\end{figure}
 
\begin{claim} \label{reduction_step_of_minttough_dp_complete_t_integer}
  Let $G$ be an arbitrary $\big\lceil (t+1)/2 \big\rceil$-connected graph. Then $G$ is \mbox{$\alpha$-critical} with $\alpha(G) = k$ if and only if $G'_{t,k}$ is minimally 1-tough.
\end{claim}
\begin{proof}
 The cases $t=1$ and $t \ge 2$ should be handled separately, but since the main steps of the proofs are similar, only the (easier) case $t=1$ is presented here.
 
 The proof of the following lemma is omitted now. (In Section~\ref{section_t>=1} a similar lemma is proved, but for a more complex graph, see Lemma~\ref{dp_complete_main_lemma_t>=1}.)
 
 \begin{lemma} \label{dp_complete_main_lemma_t=1}
  If $\alpha(G) \le k$, then $\tau(G'_{1,k})=1$.
 \end{lemma}
 
 Accepting this lemma, all we have left to show is that
 \begin{enumerate}
  \item[--] if $G$ is \mbox{$\alpha$-critical} with $\alpha(G) = k$, then $\tau(G'_{1,k} - e) < 1$ holds for any $e \in E(G)$,
  \item[--] if $\alpha(G) > k$, then $\tau(G'_{1,k}) < 1$, and
  \item[--] if either $\alpha(G) = k$ but the graph~$G$ is not \mbox{$\alpha$-critical} or $\alpha(G)<k$, then there exists an edge $e \in E(G)$ for which $\tau(G'_{1,k} - e) = 1$.
 \end{enumerate}
 
 Assume first that $G$ is \mbox{$\alpha$-critical} with $\alpha(G) = k$. Let $e \in E(G'_{1,k})$ be an arbitrary edge. If $e$ is incident to one of the vertices of~$U$, i.e., to a vertex of degree~2, then clearly $\tau(G'_{1,k} - e) < 1$. If $e$ is not incident to any of the vertices of~$U$, then it connects two vertices of~$V$. By Lemma~\ref{alpha_crit_blowup}, the subgraph $G'_{1,k} [V]$ is \mbox{$\alpha$-critical}, so in $G'_{1,k} [V] - e$ there exists an independent vertex set~$I$ of size $\alpha(G) + 1$. Let
  \[ S = (V \setminus I) \cup W \text{.} \]
 Then it is easy to see that
  \[ |S| = |V|-1 \qquad \text{and} \qquad \omega \big( (G'_{1,k}-e)-S \big) = |V| \]
 hold, so $\tau(G'_{1,k} - e) < 1$.
 
 Now assume $\alpha(G) > k$. Then let $I$ be an independent vertex set of size $\alpha(G)$ in $G'_{1,k} [V]$, and let 
  \[ S=(V \setminus I) \cup W \text{.} \]
 Then 
  \[ |S| < |V| \qquad \text{and} \qquad \omega(G'_{1,k}-S) = |V| \]
 hold, so $\tau(G'_{1,k}) < 1$.
 
 Finally, assume that either $\alpha(G) = k$ but the graph~$G$ is not \mbox{$\alpha$-critical} or $\alpha(G) < k$. Then there exists an edge $e \in E(G)$ such that $\alpha(G - e) \le k$. By Lemma~\ref{dp_complete_main_lemma_t=1}, the graph $(G-e)'_{1,k}$ is \mbox{1-tough}, but we can obtain $(G-e)'_{1,k}$ from $G'_{1,k}$ by edge-deletion, which means that $G'_{1,k}$ is not minimally \mbox{1-tough}.
\end{proof}

\begin{corollary} \label{minttough_dp_complete_t_integer}
 For any positive integer~$t$, the problem {\scshape Min-\mbox{$t$-Tough}} is DP-complete.
\end{corollary}
\begin{proof}
 In Proposition~\ref{min_t_tough_DP} we already proved that the problem {\scshape Min-\mbox{$t$-Tough}} is in DP, and it follows from Claim~\ref{reduction_step_of_minttough_dp_complete_t_integer} that we can reduce {\scshape \mbox{$\alpha$-Critical}} to it, but for this it should be also noted that $G'_{t,k}$ can be constructed from~$G$ in polynomial time.
\end{proof}

The above construction works only in the case when $t$ is a positive integer for the simple reason that the sets $V_{i,j}$, $U_{i,j}$ and $W_j$ consist of $t$~vertices.

\subsection{On the case of minimally $1/b$-tough graphs, where $b \ge 2$ is an integer} \label{subsection_1/b}


Up to this point, we only handled the case when $t$ is a positive integer. To prove Theorem~\ref{minttough_dp_complete} for the noninteger cases, we modify the previous constructions and here we illustrate these modifications with the following simple example.

Let $b \ge 2$ be an integer, let $t=1/b$, let $G$ be an arbitrary connected graph, and let $G_t$ be defined as follows. Add $b-1$ independent vertices for each original vertex $v \in V(G)$ to the graph~$G$, and connect them to~$v$ (see Figure~\ref{Fig:minttough_deterministic_constr_t=1/b}). (For an example see Figure~\ref{Fig:G1/2}.)

\begin{figure}[H]
\begin{center}
\begin{tikzpicture}
 \tikzstyle{vertex}=[draw,circle,fill=black,minimum size=3,inner sep=0]
 
 \draw (0,0) ellipse (0.4 and 2);
 \node at (-0.5,2) {$G$};
 
 \node[vertex] (v1) at (0,1.35) {};
 \node[vertex] (v2) at (0,0.45) {};
 \draw[fill] (0,-0.35) circle (0.3pt);
 \draw[fill] (0,-0.45) circle (0.3pt);
 \draw[fill] (0,-0.55) circle (0.3pt);
 \node[vertex] (vn) at (0,-1.35) {};
 
 \node at (-0.65,1.35) {$v_1$};
 \node at (-0.65,0.45) {$v_2$};
 \node at (-0.65,-1.35) {$v_n$};
 
 \node[vertex] (4) at (1,1.6) {};
 \draw[fill] (1,1.45) circle (0.3pt);
 \draw[fill] (1,1.35) circle (0.3pt);
 \draw[fill] (1,1.25) circle (0.3pt);
 \node[vertex] (5) at (1,1.1) {};
 \node[vertex] (6) at (1,0.7) {};
 \draw[fill] (1,0.55) circle (0.3pt);
 \draw[fill] (1,0.45) circle (0.3pt);
 \draw[fill] (1,0.35) circle (0.3pt);
 \node[vertex] (7) at (1,0.2) {};
 \node[vertex] (8) at (1,-1.1) {};
 \draw[fill] (1,-1.25) circle (0.3pt);
 \draw[fill] (1,-1.35) circle (0.3pt);
 \draw[fill] (1,-1.45) circle (0.3pt); 
 \node[vertex] (9) at (1,-1.6) {};
 
 \draw (v1) -- (4);
 \draw (v1) -- (5);
 \draw (v2) -- (6);
 \draw (v2) -- (7);
 \draw (vn) -- (8);
 \draw (vn) -- (9);
 
 \draw [decorate,decoration={brace,amplitude=5pt}] (1.25,1.7) -- (1.25,1) node [pos=0.5, xshift=22pt] {$b-1$};
 \draw [decorate,decoration={brace,amplitude=5pt}] (1.25,0.8) -- (1.25,0.1) node [pos=0.5, xshift=22pt] {$b-1$};
 \draw [decorate,decoration={brace,amplitude=5pt}] (1.25,-1) -- (1.25,-1.7) node [pos=0.5, xshift=22pt] {$b-1$};
\end{tikzpicture}
\caption{The graph~$G_t$ when $t=1/b$, where $b \ge 2$ is an integer.} \label{Fig:minttough_deterministic_constr_t=1/b}
\end{center}
\end{figure}

\begin{claim} \label{reduction_step_of_minttough_dp_complete_t=1/b}
  Let $G$ be an arbitrary connected graph. Then $G_t$ is minimally \mbox{$t$-tough} if and only if $G$ is almost minimally \mbox{1-tough}.
\end{claim}

Similarly as before, we can conclude the following.

\begin{corollary} \label{minttough_dp_comple_t=1/b}
 For every integer $b \ge 2$, {\scshape Min-$1/b$-Tough} is DP-complete.
\end{corollary}

In Section~\ref{section_t<=1/2}, this latter idea is extended to the case when $t<=1/2$ by ``gluing'' some other graph to the vertices of the original graph~$G$. (See Figure~\ref{Fig:G2/5}.) It is worth noting that in the case when $t=1/b$ for some integer $b \ge 2$, the obtained graph in Section~\ref{section_t<=1/2} is exactly the same as the graph~$G_t$ constructed here. After this ``gluing'', the vertices of~$G$ become cut-vertices in the obtained graph~$G_t$, thus the toughness of~$G_t$ can be at most~1/2. The plan for the cases when $t>1/2$ is to perform this so called ``gluing'' by identifying not only one, but $2t$~vertices of a smaller and a larger graph, where the larger graph resembles a minimally $\lceil t \rceil$-tough graph and the ``gluing'' procedure aims to decrease its toughness to the desired value $t$. In fact, in Sections~\ref{section_1/2<t<1} and \ref{section_t>=1} this larger graph is chosen to be a slight modification of~$G'_{\lceil t \rceil,k}$.

\section{Minimally $t$-tough graphs, where $1/2 < t < 1$} \label{section_1/2<t<1}

Before proving Theorem~\ref{minttough_dp_complete} for any positive rational number $1/2 < t < 1$, we need some preparation. First, we construct some auxiliary graphs.

\subsection{The auxiliary graph $H^{**}_{t,k}$ when $1/2 < t < 1$}

Let $t$ be a rational number such that $1/2 < t < 1$. Let $a,b$ be relatively prime positive integers such that $t = a/b$. Let $k$ be a positive integer, and let 
 \[ W = \{ w_1, \ldots, w_{a k} \} \qquad \text{and} \qquad W' = \big\{ w'_1, \ldots, w'_{(b-1) k} \big\} \text{.} \]
Place a clique on the vertices of~$W$ and a complete bipartite graph on $(W;W')$. Obviously, the toughness of this complete split graph is $a/(b-1) > t$. Deleting an edge may decrease the toughness, and now we delete edges incident to~$W'$ until the toughness remains at least~$t$ but the deletion of any other such edge would result in a graph with toughness less than~$t$. Let $H^*_{t,k}$ denote the obtained split graph. Then $\tau(H^*_{t,k}) \ge t$, and $\tau(H^*_{t,k} - e) < t$ for any edge $e \in E(H^*_{t,k})$ incident to~$W'$, i.e. there exists a vertex set $S=S(e) \subseteq W$ whose removal disconnects $H^*_{t,k}-e$ and
\[ \omega \big( (H^{*}_{t,k} - e) - S \big) > \frac{|S|}{t} \text{.} \]
Now delete all the edges induced by~$W$, and let $H^{**}_{t,k}$ denote the obtained bipartite graph.

\subsection{The auxiliary graph $H''_t$ when $1/2 < t < 1$}

Let $t$ be a rational number such that $1/2 < t < 1$. Let $a,b$ be relatively prime positive integers such that $t = a/b$ and let $H_t$ be constructed as follows. Let 
 \[ A = \{ v_1, v_2, \ldots, v_a \} \text{, } \quad B = \{ u_1, u_2, \ldots, u_b \} \text{.} \]
For any $i \in [a]$ and $j \in [b-1]$ connect $v_i$ to~$u_j$, and connect $u_b$ to~$v_1$ and~$v_a$. (In other words, $H_t$ can be obtained from the complete bipartite graph~$K_{a,b}$ by deleting $a-2$~edges incident to one vertex of the color class of size~$b$. See Figure~\ref{Fig:Ht_1/2<t<1}.)

\begin{figure}[H] 
\begin{center}
\begin{tikzpicture}
 \tikzstyle{vertex}=[draw,circle,fill=black,minimum size=4,inner sep=0]
 
 \draw[dashed] (0,0) ellipse (0.3 and 1.1);
 \node at (0,1.5) {$\overline{K}_a$};
 
 \node[vertex] (v1) at (0,0.75) [label={[xshift=-12pt, yshift=-8pt] $v_1$}] {};
 \node[vertex] (v2) at (0,0.375) {};
 \draw[fill] (0,0.1) circle (0.3pt);
 \draw[fill] (0,0) circle (0.3pt);
 \draw[fill] (0,-0.1) circle (0.3pt);
 \node[vertex] (va-1) at (0,-0.375)  {};
 \node[vertex] (va) at (0,-0.75) [label={[xshift=-12pt, yshift=-13pt] $v_a$}] {};
 
 \node at (0,-1.5) {$A$};
 
 \node[vertex] (ub) at (-1,0) [label=left:{$u_b$}] {};
 
 \draw[dashed] (2,0) ellipse (0.4 and 1.5);
 
 \node at (2,2) {$\overline{K}_{b-1}$};
 \node at (2,-2) {$B'$};
 
 \node[vertex] (u1) at (2,1) [label={[xshift=18pt, yshift=-10pt] $u_1$}] {};
 \node[vertex] (u2) at (2,0.5) [label={[xshift=18pt, yshift=-10pt] $u_2$}] {};
 \draw[fill] (2,0.1) circle (0.3pt);
 \draw[fill] (2,0) circle (0.3pt);
 \draw[fill] (2,-0.1) circle (0.3pt);
 \node[vertex] (ub-2) at (2,-0.5) [label={[xshift=23pt, yshift=-12pt] $u_{b-2}$}] {};
 \node[vertex] (ub-1) at (2,-1)  [label={[xshift=23pt, yshift=-12pt] $u_{b-1}$}] {};
 
 \draw (ub) -- (v1);
 \draw (ub) -- (va);
 
 \draw (0.05,1.2) -- (1.9,1.575);
 \draw (0.05,-1.2) -- (1.9,-1.575);
 \draw (0.2,1.1) -- (1.7,-1.4);
 \draw (0.2,-1.1) -- (1.7,1.4);
\end{tikzpicture}
\caption{The graph~$H_t$, when $1/2 < t < 1$.}
\label{Fig:Ht_1/2<t<1}
\end{center}
\end{figure}

\begin{claim} \label{properties_of_Ht_1/2<t<1}
 Let $t$ be a rational number such that $1/2 < t < 1$. Then $\tau(H_t) = t$.
\end{claim}
\begin{proof}
 Let $B' = B \setminus \{ u_b \}$ and let $S$ be an arbitrary cutset in~$H_t$. Now we show that $\omega(H_t - S) \le |S|/t$.
 
 \bigskip
 
 \textit{Case 1:} $A \subseteq S$.
 
 Then $|S| \ge a$ and $\omega(H_t-S) \le b$. Since $t = a/b < 1$, it follows that
 \[ \omega(H_t - S) \le b = \frac{a}{t} \le \frac{|S|}{t} \text{.} \]
 
 \bigskip
 
 \textit{Case 2:} $B' \subseteq S$.
 
 If $u_b \in S$ as well, then $|S| \ge b$ and $\omega(H_t-S) \le a$. Since $t = a/b < 1$, it follows that
 \[ \omega(H_t - S) \le a = bt < \frac{b}{t} \le \frac{|S|}{t} \text{.} \]
 If $u_b \notin S$, then $|S| \ge b-1$ and $\omega(H_t-S) \le a-1$. Since $t = a/b < 1$, it follows that
 \[ \omega(H_t - S) \le a-1 \le \frac{b-1}{t} \le \frac{|S|}{t} \text{.} \]
 
 \bigskip
 
 \textit{Case 3:} $A \nsubseteq S$ and $B' \nsubseteq S$.
 
 Then $\omega(H_t-S) \le 2$, but since $S$ is a cutset, $\omega(H_t-S) = 2$. Obviously, there is no cut-vertex in~$H_t$, thus $|S| \ge 2$. Since $t<1$, it follows that
 \[ \omega(H_t - S) = 2 < \frac{2}{t} \le \frac{|S|}{t} \text{.} \]
 
 \bigskip
 
 Hence $\tau(H_t) \ge t$. On the other hand, the vertex set $S=A$ is a cutset in~$H_t$ with $|S| = a$ and $\omega(H_t-S) = b$, so $\tau(H_t) \le t$.
 
 Therefore, $\tau(H_t) = t$.
\end{proof}

By repeatedly deleting some edges of~$H_t$, eventually we obtain a minimally \mbox{$t$-tough} graph, let us denote it with $H'_t$ (i.e. if there exists an edge whose deletion does not decrease the toughness, then we delete it). Obviously, we could not delete the edges incident to~$u_b$, so the vertex~$u_b$ still has degree~2. Let $e$ denote the edge connecting $v_1$ and~$u_b$ and let $H_t'' = H'_t - e$. Note that $H_t''$ is a bipartite graph with color classes $A$ and $B$.

\subsection{The proof of Theorem~\ref{minttough_dp_complete} when $1/2 < t < 1$} \label{subsection_1/2<t<1}

\begin{thm} \label{minttough_dp_complete_1/2<t<1}
 For any rational number $t$ with $1/2 < t < 1$, the problem {\scshape Min-\mbox{$t$-Tough}} is DP-complete.
\end{thm}

\begin{proof}
 Let $t$ be a rational number such that $1/2 < t < 1$. In Proposition~\ref{min_t_tough_DP} we already proved that the problem {\scshape Min-\mbox{$t$-Tough}} is in DP. To show that it is DP-hard, we reduce {\scshape \mbox{$\alpha$-Critical}} to it.

 Let $a,b$ be relatively prime positive integers such that $t = a/b$, let $G$ be an arbitrary \mbox{2-}con\-nected graph on the vertices $v_1, \ldots, v_n$ and let $G_{t,k}$ be defined as follows. For all $i \in [n]$ let
 \[ V_i = \big\{ v_{i,j} \bigm| i \in [n], j \in [a k] \big\} \]
 and place a clique on the vertices of~$V_i$. For all $i_1, i_2 \in [n]$ if $v_{i_1} v_{i_2} \in E(G)$, then place a complete bipartite graph on $(V_{i_1}; V_{i_2})$. (This subgraph is denoted by $\tilde{G}$ in Figure~\ref{minttough_deterministic_constr_1/2<t<1}.) For all $i \in [n], j \in [a k]$ ``glue'' the graph~$H_t''$ to the vertex~$v_{i,j}$ by identifying $v_{i,j}$ with the vertex~$v_1$ of~$H_t''$ and let $H^{i,j}$ denote the $(i,j)$-th copy of $H_t''$ and let $A^{i,j}$ denote the $(i,j)$-th copy of its color class $A$, and let~$v'_{i,j}$ and~$u_{i,j}$ denote the $(i,j)$-th copies of the vertices~$v_a$ and~$u_b$, respectively. Let
  \[ V = \bigcup_{i=1}^n V_i \text{,} \]
  \[ V' = \big\{ v'_{i,j} \bigm| i \in [n], j \in [a k] \big\} \]
 and
  \[ U = \big\{ u_{i,j} \bigm| i \in [n], j \in [a k] \big\} \text{.} \]
 Add the vertex sets
  \[ W = \big\{ w_j \bigm| j \in [a k] \big\} \]
 and
  \[ W' = \big\{ w'_1, \ldots, w'_{(b-1) k} \big\} \]
 to the graph and place the bipartite graph $H^{**}_{t,k}$ on $(W;W')$. For all $i \in [n]$ and $j \in [a k]$ connect $w_j$ to~$u_{i,j}$. See Figure~\ref{minttough_deterministic_constr_1/2<t<1}. (For an example see Figure~\ref{Fig:G2/3,1}.) Now $k$ is part of the input of the problem {\scshape \mbox{$\alpha$-Critical}}, therefore the graph $H^{**}_{t,k}$ must be constructed in polynomial time and by Theorem~\ref{split_general_thm}, this can be done. On the other hand, $t$~is not part of the input of the problem {\scshape Min-\mbox{$t$-Tough}}, therefore the graph~$H''_t$ can be constructed in advance. Hence, $G_{t,k}$ can be constructed from~$G$ in polynomial time.

 \begin{figure}[H]
 \begin{center}
 \begin{tikzpicture}[scale=1.25]
  \tikzstyle{vertex}=[draw,circle,fill=black,minimum size=3,inner sep=0]
   
  \draw (-0.3,0) ellipse (1.15 and 3);
  \draw (0,1.95) ellipse (0.5 and 0.5);
  \draw (0,0.65) ellipse (0.5 and 0.5);
  \draw[fill] (0,-0.55) circle (0.3pt);
  \draw[fill] (0,-0.65) circle (0.3pt);
  \draw[fill] (0,-0.75) circle (0.3pt);
  \draw (0,-1.95) ellipse (0.5 and 0.5);
  
  \node at (-1,3) {$\tilde{G}$};
  \node at (-0.2,1.95) {$V_1$};
  \node at (-0.2,0.65) {$V_2$};
  \node at (-0.2,-1.95) {$V_n$};
  
  \node[vertex] (v11) at (0,2.3) [label={[xshift=-23pt, yshift=-10pt] $v_{1,1}$}] {};
  \draw[fill] (0,2.05) circle (0.3pt);
  \draw[fill] (0,1.95) circle (0.3pt);
  \draw[fill] (0,1.85) circle (0.3pt);
  \node[vertex] (v1a) at (0,1.6) [label={[xshift=-25pt, yshift=-10pt] $v_{1,ak}$}] {};
  \node[vertex] (v21) at (0,1) [label={[xshift=-23pt, yshift=-10pt] $v_{2,1}$}] {};
  \draw[fill] (0,0.75) circle (0.3pt);
  \draw[fill] (0,0.65) circle (0.3pt);
  \draw[fill] (0,0.55) circle (0.3pt);
  \node[vertex] (v2a) at (0,0.3) [label={[xshift=-25pt, yshift=-10pt] $v_{2,ak}$}] {};
  \node[vertex] (vn1) at (0,-1.6) [label={[xshift=-23pt, yshift=-10pt] $v_{n,1}$}] {};
  \draw[fill] (0,-1.85) circle (0.3pt);
  \draw[fill] (0,-1.95) circle (0.3pt);
  \draw[fill] (0,-2.05) circle (0.3pt);
  \node[vertex] (vna) at (0,-2.3) [label={[xshift=-25pt, yshift=-10pt] $v_{n,ak}$}] {};
  
  \draw (0.75,2.3) ellipse (1 and 0.15);
  \draw (0.75,1.6) ellipse (1 and 0.15);
  \draw (0.75,1) ellipse (1 and 0.15);
  \draw (0.75,0.3) ellipse (1 and 0.15);
  \draw (0.75,-1.6) ellipse (1 and 0.15);
  \draw (0.75,-2.3) ellipse (1 and 0.15);
  
  \node at (2.125,2.3) {$H^{1,1}$};
  \node at (2.2,-2.3) {$H^{n,ak}$};
  
  \node[vertex] (v'11) at (1,2.3) [label={[xshift=0pt, yshift=0pt] $v'_{1,1}$}] {};
  \node[vertex] (v'1a) at (1,1.6) {};
  \node[vertex] (v'21) at (1,1) {};
  \node[vertex] (v'2a) at (1,0.3) {};
  \node[vertex] (v'n1) at (1,-1.6) {};
  \node[vertex] (v'na) at (1,-2.3) [label={[xshift=-5pt, yshift=-25pt] $v'_{n,a k}$}] {};
  
  \node[vertex] (u11) at (1.5,2.3) [label={[xshift=0pt, yshift=0pt] $u_{1,1}$}] {};
  \node[vertex] (u1a) at (1.5,1.6) {};
  \node[vertex] (u21) at (1.5,1) {};
  \node[vertex] (u2a) at (1.5,0.3) {};
  \node[vertex] (un1) at (1.5,-1.6) {};
  \node[vertex] (una) at (1.5,-2.3) [label={[xshift=5pt, yshift=-25pt] $u_{n,a k}$}] {};
  
  \draw (u11) -- (v'11);
  \draw (u1a) -- (v'1a);
  \draw (u21) -- (v'21);
  \draw (u2a) -- (v'2a);
  \draw (un1) -- (v'n1);
  \draw (una) -- (v'na);
  
  \node[vertex] (w1) at (2.5,0.75) [label={[xshift=10pt, yshift=-10pt] $w_1$}] {};
  \draw[fill] (2.5,0.1) circle (0.3pt);
  \draw[fill] (2.5,0) circle (0.3pt);
  \draw[fill] (2.5,-0.1) circle (0.3pt);
  \node[vertex] (wa) at (2.5,-0.75) [label={[xshift=12pt, yshift=-10pt] $w_{a k}$}] {};
  
  \node at (2.5,1.75) {$W$};
  
  \draw (u11) -- (w1);
  \draw (u1a) -- (wa);
  \draw (u21) -- (w1);
  \draw (u2a) -- (wa);
  \draw (un1) -- (w1);
  \draw (una) -- (wa);
  
  \node[vertex] (w'1) at (3.5,1) [label={[xshift=25pt, yshift=-12pt] $w'_1$}] {};
  \draw[fill] (3.5,0.1) circle (0.3pt);
  \draw[fill] (3.5,0) circle (0.3pt);
  \draw[fill] (3.5,-0.1) circle (0.3pt);
  \node[vertex] (w'b) at (3.5,-1) [label={[xshift=35pt, yshift=-15pt] $w'_{(b-1) k}$}] {};
  
  \node at (3.75,1.75) {$W'$};
  
  \draw (3.125,0) ellipse (1 and 1.5);
  
  \node at (3.125,-1.85) {$H_{t,k}^{**}$};
 \end{tikzpicture}
 \caption{The graph $G_{t,k}$, when $1/2 < t < 1$.} \label{minttough_deterministic_constr_1/2<t<1}
 \end{center}
 \end{figure}
 
 To show that $G$ is \mbox{$\alpha$-critical} with $\alpha(G) = k$ if and only if $G_{t,k}$ is minimally \mbox{$t$-tough}, first we prove the following lemma.
 
 \begin{lemma} \label{dp_complete_main_lemma_1/2<t<1}
  Let $G$ be a \mbox{2-}con\-nected graph with $\alpha(G) \le k$. Then $G_{t, k}$ is \mbox{$t$-tough}.
 \end{lemma}
 \begin{proof}
  Let $S \subseteq V(G_{t, k})$ be a cutset in $G_{t, k}$. We need to show that $\omega(G_{t, k} - S) \le |S|/t$.
  
  First, we show that the following assumption can be made for~$S$.
  
  \begin{description}
   \item[] (1) $U \cap S = \emptyset$.
   
    Suppose that $u_{i,j} \in S$ for some $i \in [n], j \in [a k]$. If $v'_{i,j} \in S$, then after the removal of~$v'_{i,j}$, the vertex $u_{i,j}$ has degree~$1$, so there is no need to remove it. Similarly, if $w_j \in S$, then we can also assume that $u_{i,j} \notin S$. If $v'_{i,j},w_j \notin S$, then considering $S' = S \setminus \{ u_{i,j} \}$ instead of~$S$ decreases the number of components only by one, meaning that if $S'$ is a cutset in $G_{t,k}$, then it is enough to show that $\omega(G_{t,k}-S') \le |S'|/t$ since it implies
     \[ \omega(G_{t,k}-S) = \omega(G_{t,k}-S') + 1 \le \frac{|S'|}{t} + 1 = \frac{|S|-1}{t} + 1 \le \frac{|S|}{t} \text{,} \]
    where the last inequality is valid since $t<1$. If $S'$ is not a cutset in $G_{t,k}$, then $\omega(G_{t,k} - S) = 2$ and $|S| \ge 2$ since $u_{i,j}$ has degree 2 and is not a cut-vertex in $G_{t,k}$, i.e.
     \[ \omega(G_{t,k}-S) = 2 \le |S| \le \frac{|S|}{t} \text{,} \]
    where again the last inequality is valid since $t<1$. This completes the validation of assumption~(1).
  \end{description}
  
  Now there are two cases.
  
  \bigskip
  
  \textit{Case 1:} $W \subseteq S$.
  
  After the removal of~$W$, the vertices of~$W'$ are isolated; therefore we can assume that $W' \cap S = \emptyset$.
  
  To write up a formula for $|S|$ and $\omega(G_{t,k} - S)$, we need to introduce some notations. Let
   \[ C = \big\{ (i,j) \in [n] \times [a k] \bigm| v_{i,j} \in V \cap S \big\} \text{,} \]
   \[ c_{i,j} = \big| V(H^{i,j}) \cap S \big| - 1 \]
  for all $(i,j) \in C$, and 
  \[ d_{i,j} = \big| V(H^{i,j}) \cap S \big| \]
  for all $(i,j) \in \big( [n] \times [a k] \big) \setminus C$. Finally, let
   \[ D = \big\{ (i,j) \in \big( [n] \times [a k] \big) \setminus C \bigm| d_{i,j} > 0 \big\} \text{.} \]
  
  Using these notations it is clear that
   \[ |S| = \sum_{(i,j) \in [n] \times [ak]} \big| V(H^{i,j}) \cap S \big| + |W| = |C| + \sum_{(i,j) \in C} c_{i,j} + \sum_{(i,j) \in D} d_{i,j} + a k \text{.} \]
  By the assumption that $W \subseteq S$, in $G_{t, k}-S$ the $(b-1)k$ vertices of~$W'$ are isolated. Since $\alpha \big( G_{t, k}[V] \big) = \alpha(G)$, the removal of $V \cap S$ from $G_{t, k}[V]$ leaves at most $\alpha(G)$ components. By Claim~\ref{properties_of_Ht_1/2<t<1} and Proposition~\ref{obs_below1}, for any $(i,j) \in C$ the removal of $V(H^{i,j}) \cap S$ from $H^{i,j}$ leaves at most $(c_{i,j} + 1)/t$ components. By Proposition~\ref{obs_below1}, for any $(i,j) \in D$ the removal of $V(H^{i,j}) \cap S$ from $H^{i,j}$ leaves at most $d_{i,j}/t+1$ components, but the component of~$v_{i,j}$ has been already counted. Hence
  \begin{gather*} 
   \omega(G_{t, k}-S) \le (b-1) k + \alpha(G) + \sum_{(i,j) \in C} \frac{c_{i,j} + 1}{t} + \sum_{(i,j) \in D} \frac{d_{i,j}}{t} \\
   \le b k + \frac{ |C| + \sum_{(i,j) \in C} c_{i,j} + \sum_{(i,j) \in D} d_{i,j} }{t} = \frac{|S|}{t} \text{,}
  \end{gather*}
  using that $\alpha(G) \le k$.
  
  \bigskip
  
  \textit{Case 2:} $W \nsubseteq S$.
  
  Assume that $w_{j_0} \notin S$ for some $j_0 \in [ak]$. In this case, using assumption~(1), we can also assume the following.
      
  \begin{description}
   \item[] (2) There exists at most one $i \in [n]$ for which $v_{i,j_0} \in S$.
   
    Suppose that $v_{i_1,j_0}, v_{i_2,j_0} \in S$ for some $i_1,i_2 \in [n]$. By assumption~(1), the component of~$w_{j_0}$ contains all of the vertices $u_{1,j_0}, u_{2,j_0}, \ldots, u_{n,j_0}$. Now considering the cutset $S' = S \cup \{ w_{j_0} \}$ instead of~$S$ increases the number of components by at least two: it disconnects both $u_{i_1,j_0}$ and $u_{i_2,j_0}$ from the vertices $\big\{ u_{i,j_0} \bigm| i \in [n] \setminus \{i_1,i_2\} \big\}$, and it also disconnects $u_{i_1,j_0}$ from $u_{i_2,j_0}$ (and of course it can also disconnect other vertices of $\big\{ u_{i,j_0} \bigm| i \in [n] \big\}$ from each other). Then it is enough to show that $\omega(G_{t,k}-S') \le |S'|/t$ since it implies
     \[ \omega(G_{t,k} - S) \le \omega(G_{t,k} - S') - 2 \le \frac{|S'|}{t} - 2 = \frac{|S|+1}{t} - 2 < \frac{|S|}{t} \text{,} \]
    where the last inequality is valid since $t > 1/2$. Proceeding further, we can obtain a cutset $S^*$ for which $W \subseteq S^*$ holds; and such sets were already handled in Case 1.
   
   \item[] (3) $(G_{t,k} - S)[V]$ is connected.
   
    By assumption~(2), there exists at most one $i \in [n]$ for which $V_i \subseteq S$. Since~$G$ is \mbox{2-}con\-nected, this implies that $(G_{t,k} - S)[V]$ is connected.
   
   \item[] (4) There exists at most one $i \in [n]$ for which $v_{i,j_0}$ and~$u_{i,j_0}$ belong to different components in $G_{t, k}-S$.
   
    Suppose that $v_{i_1,j_0}, u_{i_1,j_0}$ belong to different components in $G_{t, k}-S$, and so do $v_{i_2,j_0}, u_{i_2,j_0}$ for some $i_1,i_2 \in [n]$. Similarly as in the proof of assumption~(2), considering the cutset $S' = S \cup \{ w_{j_0} \}$ instead of~$S$ increases the number of components by at least two, so it is enough to show that $\omega(G_{t,k}-S') \le |S'|/t$.
   
   \item[] (5) In $G_{t, k}-S$ all the remaining vertices of $\big\{ v_{i,j_0}, u_{i,j_0} \bigm| i \in [n] \big\}$ belong to the component of~$w_{j_0}$.
   
    It follows directly from assumptions~(1), (2) and~(3).
   
   \item[] (6) In $G_{t, k}-S$ all the remaining vertices of $V$ belong to the component of~$w_{j_0}$.
   
    It follows directly from assumptions~(3) and~(5).
   
   \item[] (7) In $G_{t, k}-S$ all the remaining vertices of $V \cup W$ belong to the same component.
   
    It follows directly from assumptions~(5) and~(6).
  \end{description}
   
  By assumption~(7), in $G_{t,k}-S$ there is a component containing all the remaining vertices of $V \cup W$, and every other component is either an isolated vertex of~$W'$ (since $G_{t,k}[W \cup W']$ is a bipartite graph) or a component of $H^{i,j}-\big( V(H^{i,j}) \cap S \big)$ for some $i \in [n], j \in [a k]$. Hence we can also assume the following.
   
  \begin{description}
   \item[] (8) $W' \cap S = \emptyset$.
  \end{description}
   
  By assumption~(5) and Proposition~\ref{obs_below1} and the properties of $H_{t,k}^{**}$, the removal of $W \cap S$ from $H_{t,k}^{**}$ leaves at most $|W \cap S|/t$ components, but the component of~$w_{j_0}$ has been already counted.
   
  Using the previous notations,
   \[ |S| = |C| + \sum_{(i,j) \in C} c_{i,j} + \sum_{(i,j) \in D} d_{i,j} + |W \cap S| \]
  and
  \begin{gather*} 
   \omega(G_{t, k}-S) \le 1 + \sum_{(i,j) \in C} \frac{c_{i,j} + 1}{t} + \sum_{(i,j) \in D} \frac{d_{i,j}}{t} + \left( \frac{|W \cap S|}{t} - 1 \right) \\
   = \frac{ |C| + \sum_{(i,j) \in C} c_{i,j} + \sum_{(i,j) \in D} d_{i,j} }{t} + \frac{|W \cap S|}{t} = \frac{|S|}{t} \text{.}
  \end{gather*}
  
  \bigskip
  
  This means that $\tau(G_{t, k}) \ge t$.
 \end{proof}
 
 Now we return to the proof of Theorem~\ref{minttough_dp_complete_1/2<t<1} and we show that $G$ is \mbox{$\alpha$-critical} with $\alpha(G) = k$ if and only if $G_{t,k}$ is minimally \mbox{$t$-tough}.
 
 Let us assume that $G$ is \mbox{$\alpha$-critical} with $\alpha(G) = k$. By Lemma~\ref{dp_complete_main_lemma_1/2<t<1}, the graph $G_{t,k}$ is \mbox{$t$-tough}, i.e.~$\tau(G_{t,k}) \ge t$.
 
 Let $I$ be an independent vertex set of size $\alpha(G)$ in $G_{t,k}[V]$.
 
 Recall the definition of $A^{i,j}$ from the beginning of the proof: it is the color class $A$ in the corresponding copy of $H''_t$. Let
  \[ J = \big\{ (i,j) \in [n] \times [ak] \bigm| v_{i,j} \in I \big\} \]
 and
  \[ S = \left( \bigcup_{(i,j) \notin J} A^{i,j} \right) \cup W \text{.} \]
 Then $S$ is a cutset in $G_{t,k}$ with
  \[|S| = a \big( |V| - \alpha(G) \big) + ak = a |V| \]
 and
  \[ \omega (G_{t, k} - S) = \alpha(G) + b \big( |V| - \alpha(G) \big) + (b-1)k = b |V| = \frac{|S|}{t} \text{,} \]
 so $\tau(G_{t, k}) \le t$.
 
 Therefore, $\tau(G_{t,k}) = t$.
 
 Let $e \in E(G_{t, k})$ be an arbitrary edge. We need to show that $\tau(G_{t,k} - e) < t$. Now we have four cases.
 
 \bigskip
 
 \textit{Case 1:} $e$ has an endpoint in~$U$.
 
 Then this endpoint has degree~2, so $\tau(G_{t, k} - e) \le 1/2 < t$.
 
 \bigskip
 
 \textit{Case 2:} $e$ has an endpoint in~$W'$.
 
 By the properties of $H^*_{t,k}$, there exists a cutset $S \subseteq W$ in $H_{t,k}^*-e$ for which 
  \[ \omega \big( (H^*_{t,k} - e) - S \big) > \frac{|S|}{t} \text{.} \]
 Note that $S$ is also a cutset in $G_{t, k}-e$ and
  \[ \omega \big( (G_{t, k} - e) - S \big) > \frac{|S|}{t} \text{,} \]
 so $\tau(G_{t,k} - e) < t$.

 \bigskip
 
 \textit{Case 3:} $e$ is induced by $H^{i_0,j_0}$ for some $i_0 \in [n], j_0 \in [a k]$.
 
 By Proposition~\ref{minttoughlemma}, there exists a vertex set $S \subseteq V(H'_t)$ for which 
  \[ \omega \big( (H'_t - e) - S \big) > \frac{|S|}{t} \text{.} \]
 Consider the $(i_0,j_0)$-th copy of the vertex set~$S$ in $G_{t,k}-e$; let us denote it with $S_{i_0, j_0}$. If $v_{i_0,j_0} \in S_{i_0, j_0}$, then $S_{i_0, j_0}$ is a cutset in $G_{t,k}-e$ and 
  \[ \omega \big( (G_{t,k} - e) - S_{i_0, j_0} \big) = \omega \big( (H'_t - e) - S \big) > \frac{|S|}{t} \text{,} \]
 so $\tau(G_{t,k} - e) < t$. Now assume that $v_{i_0,j_0} \notin S_{i_0, j_0}$. Let $I$ be an independent vertex set of size $\alpha(G)$ in $G_{t,k}[V]$ that contains~$v_{i_0,j_0}$ (by Proposition~\ref{alphacrit_max_indep}, such an independent vertex set exists). Let
  \[ J = \big\{ (i,j) \in [n] \times [ak] \bigm| v_{i,j} \in I \big\} \]
 and
  \[ S' = S_{i_0, j_0} \cup \left( \bigcup_{(i,j) \notin J} A^{i,j} \right) \cup W \text{.} \]
 Then $S'$ is a cutset in $G_{t,k}-e$ with
  \[ |S'| = |S| + a \big( |V| - \alpha(G) \big) + a k = |S| + a|V| \]
 and
  \[ \omega \big( (G_{t, k}-e)-S' \big) > \frac{|S|}{t} + \alpha(G) + b \big( |V| - \alpha(G) \big) + (b-1) k = \frac{|S|}{t} + b|V| = \frac{|S'|}{t} \text{,} \]
 so $\tau(G_{t, k} - e) < t$.
 
 \bigskip
 
 \textit{Case 4:} $e$ connects two vertices of~$V$.
 
 By Lemma~\ref{alpha_crit_blowup}, the graph $G_{t, k}[V]$ is \mbox{$\alpha$-critical}, so in $(G_{t, k}-e)[V]$ there exists an independent vertex set~$I$ of size $\alpha(G) + 1$. Let
  \[ J = \big\{ (i,j) \in [n] \times [ak] \bigm| v_{i,j} \in I \big\} \]
 and
  \[ S = \left( \bigcup_{(i,j) \notin J} A^{i,j} \right) \cup W \text{.} \]
 Then $S$ is a cutset in $G_{t,k}-e$ with
  \[ |S| = a \big( |V| - \alpha(G)-1 \big) + a k = a|V|-a \]
 and
  \[ \omega \big( (G_{t, k}-e)-S \big) = \alpha(G) + 1 + b \big( |V| - \alpha(G) - 1 \big) + (b-1) k = b|V| - b + 1 > \frac{|S|}{t} \text{,} \]
 so $\tau(G_{t, k} - e) < t$.

 \bigskip
 
 Therefore, if $G$ is \mbox{$\alpha$-critical} with $\alpha(G) = k$, then $G_{t,k}$ is minimally \mbox{$t$-tough}.
 
 \bigskip
 
 Now let us assume that $G$ is not \mbox{$\alpha$-critical} with $\alpha(G) = k$, i.e.~either $\alpha(G) \ne k$ or even though $\alpha(G) = k$, the graph~$G$ is not \mbox{$\alpha$-critical}.
 
 \bigskip
 
 \textit{Case I:} $\alpha(G) > k$.
 
 Let $I$ be an independent vertex set of size $\alpha(G)$ in $G_{t, k}[V]$ and let
  \[ J = \big\{ (i,j) \in [n] \times [ak] \bigm| v_{i,j} \in I \big\} \]
 and
  \[ S = \left( \bigcup_{(i,j) \notin J} A^{i,j} \right) \cup W \text{.} \]
 Then $S$ is a cutset in $G_{t,k}-e$ with
  \[ |S| = a \big( |V| - \alpha(G) \big) + a k = a |V| - a \big( \alpha(G) - k \big) \]
 and
 \begin{gather*}
  \omega(G_{t, k}-S) = \alpha(G) + b \big( |V| - \alpha(G) \big) + (b-1) k = b |V| - (b-1) \big( \alpha(G) - k \big) \\
  > b|V| - b \big( \alpha(G) - k \big) = \frac{|S|}{t} \text{,} 
 \end{gather*}
 so $\tau(G_{t,k}) < t$, which means that $G_{t,k}$ is not minimally \mbox{$t$-tough}.
 
 \bigskip
 
 \textit{Case II:} $\alpha(G) \le k$.
 
 Since $G$ is not \mbox{$\alpha$-critical} with $\alpha(G) = k$, there exists an edge $e \in E(G)$ such that $\alpha(G - e) \le k$. By Lemma~\ref{dp_complete_main_lemma_1/2<t<1}, the graph $(G-e)_{t, k}$ is \mbox{$t$-tough}, but it can be obtained from $G_{t, k}$ by edge-deletion, which means that $G_{t, k}$ is not minimally \mbox{$t$-tough}.
\end{proof}

\section{Minimally $t$-tough graphs, where $t \ge 1$} \label{section_t>=1}

This whole section resembles the previous one in structure. However, it requires some additional ideas that make the proofs more complicated. First, again, we construct some auxiliary graphs.

Let $t \ge 1$ be a rational number. It is easy to see that either $\lceil 2t \rceil = 2 \lceil t \rceil$ or $\lceil 2t \rceil = 2 \lceil t \rceil - 1$. Let $T = \lceil t \rceil$,
 \[ T' = \lceil 2t \rceil - \lceil t \rceil = \begin{cases}
                                               T   & \text{if $\lceil 2t \rceil = 2 \lceil t \rceil$,} \\
                                               T-1 & \text{if $\lceil 2t \rceil = 2 \lceil t \rceil - 1$,}
                                              \end{cases} \]
and
 \[ M = \left\lceil \frac{2 \lceil t \rceil}{\lceil 2t \rceil} \right\rceil = \begin{cases}
                                                                               1 & \text{if $\lceil 2t \rceil = 2 \lceil t \rceil$,} \\
                                                                               2 & \text{if $\lceil 2t \rceil = 2 \lceil t \rceil - 1$.}
                                                                              \end{cases} \]
Let $a, b$ be the smallest positive integers such that $b \ge 3$ and $t = a/b$.

\subsection{The auxiliary graph $H^{**}_{t,k}$ when $t \ge 1$}

Let $k$ be a positive integer that is divisible by~$a$. Note that in this case 
\[ \left( \frac{MT'}{t}-1 \right) k = \begin{cases}
                                       \frac{Tbk}{a} - k      & \text{if $\lceil 2t \rceil = 2 \lceil t \rceil$,} \\
                                       \frac{2(T-1)bk}{a} - k & \text{if $\lceil 2t \rceil = 2 \lceil t \rceil - 1$}
                                      \end{cases} \]
is a positive integer. Let 
 \[ W = \big\{ w_{j,l,m} \bigm| j \in [k], l \in [T'], m \in M \big\} \]
and 
 \[ W' = \big\{ w'_1, \ldots, w'_{(MT'/t-1) k} \big\} \text{.} \]
Place a clique on the vertices of $W$ and a complete bipartite graph on $(W;W')$. Obviously, the toughness of this complete split graph is 
 \[ \frac{kMT'}{(MT'/t-1)k} = \frac{1}{\frac{1}{t} - \frac{1}{MT'}} > t \text{.} \]
Deleting an edge may decrease the toughness, and now we delete edges incident to~$W'$ until the toughness remains at least~$t$ but the deletion of any other such edge would result in a graph with toughness less than $t$. Let $H^*_{t,k}$ denote the obtained split graph. Then $\tau(H^*_{t,k}) \ge t$, and $\tau(H^*_{t,k} - e) < t$ for any edge $e \in E(H^*_{t,k})$ incident to~$W'$, i.e.~there exists a vertex set $S=S(e) \subseteq W$ whose removal disconnects $H^*_{t,k}-e$ and
 \[ \omega \big( (H^{*}_{t,k} - e) - S \big) > \frac{|S|}{t} \text{.} \]
Now delete all the edges induced by~$W$, and let $H^{**}_{t,k}$ denote the obtained bipartite graph.

\subsection{The auxiliary graph $H''_{t,k}$ when $t \ge 1$}

Let $H_t$ be constructed as follows. Let
 \[ V'_1 = \{ v'_1, \ldots, v'_T \} \text{,} \qquad V'_2 = \{ v'_{T+1}, \ldots, v'_{2T} \} \text{,} \qquad V'_3 = \{ v'_{2T + 1}, \ldots, v'_{aT} \} \text{,} \]
 \[ V'' = \{ v''_1, \ldots, v''_T \} \text{,} \]
 \[ U'_1 = \{ u'_1, \ldots, u'_T \} \text{,} \qquad U'_2 = \{ u'_{T+1}, \ldots, u'_{2T} \} \text{,} \qquad U'_3 = \{ u'_{2T + 1}, \ldots, u'_{bT-1} \} \text{,} \]
 \[ U'' = \{ u''_1, \ldots, u''_{T'} \} \text{,} \]
and
 \[ U''_1 = \{ u''_1, \ldots, u''_T \} \text{.} \]
Place a clique on the vertices of $V'_1$, $V'_2$, $V'_3$, and $U''$. For all $l \in [T]$ connect $v''_l$ to~$v'_l$ and to~$u'_l$, and connect $v'_{T+l}$ to~$u'_{T+l}$. Connect all the vertices of~$V'_3$ to all the vertices of $V'_1 \cup V'' \cup U'_1 \cup U'_2$, and connect all the vertices of~$V'_2$ to all the vertices of~$U''$. Finally, add a new vertex~$x$ to the graph and connect it to all the vertices of $V'_1 \cup U''$. See Figure~\ref{Fig:Ht_t>=1}.

\begin{figure}[H]
 \begin{center}
 \begin{tikzpicture}
  \tikzstyle{vertex}=[draw,circle,fill=black,minimum size=4,inner sep=0]

  \begin{scope}[shift={(-0.5,0)}]
  \node[vertex] (x) at (0,1.5) [label={[xshift=0pt, yshift=0pt] $x$}] {};
  
  \draw (0,0) ellipse (1 and 0.25);
  \node[vertex] (u''1) at (-0.75,0) [label={[xshift=-2pt, yshift=0pt] $u''_{1}$}] {};
  \node[vertex] (u''2) at (-0.25,0) {};
  \draw[fill] (0.15,0) circle (0.3pt);
  \draw[fill] (0.25,0) circle (0.3pt);
  \draw[fill] (0.35,0) circle (0.3pt);
  \node[vertex] (u''T) at (0.75,0) [label={[xshift=4pt, yshift=0pt] $u''_{T'}$}] {};
  \node at (-1.75,0) {$U''$};
  
  \draw (0,-2) ellipse (1 and 0.25);
  \node[vertex] (v'T+1) at (-0.75,-2) [label={[xshift=0pt, yshift=-23pt] $v'_{T+1}$}] {};
  \node[vertex] (v'T+2) at (-0.25,-2) {};
  \draw[fill] (0.15,-2) circle (0.3pt);
  \draw[fill] (0.25,-2) circle (0.3pt);
  \draw[fill] (0.35,-2) circle (0.3pt);
  \node[vertex] (v'2T) at (0.75,-2) [label={[xshift=0pt, yshift=-22pt] $v'_{2T}$}] {};
  \node at (-1.75,-2) {$V'_2$};
  
  \draw (x) -- (u''1);
  \draw (x) -- (u''2);
  \draw (x) -- (u''T);
  
  \coordinate (A) at (-1,-0.1);
  \coordinate (B) at (1,-0.1);
  \coordinate (C) at (-1,-1.9);
  \coordinate (D) at (1,-1.9);
  
  \draw ($(A)!0.05!(C)$) -- ($(A)!0.95!(C)$);
  \draw ($(B)!0.05!(D)$) -- ($(B)!0.95!(D)$);
  \draw ($(A)!0.05!(D)$) -- ($(A)!0.95!(D)$);
  \draw ($(B)!0.05!(C)$) -- ($(B)!0.95!(C)$);
  
  \end{scope}
  
  \draw (4,0) ellipse (1 and 0.25);
  \node[vertex] (v'1) at (3.25,0) [label={[xshift=-7pt, yshift=-19pt] $v'_{1}$}] {};
  \node[vertex] (v'2) at (3.75,0) {};
  \draw[fill] (4.15,0) circle (0.3pt);
  \draw[fill] (4.25,0) circle (0.3pt);
  \draw[fill] (4.35,0) circle (0.3pt);
  \node[vertex] (v'T) at (4.75,0) [label={[xshift=7pt, yshift=-19pt] $v'_{T}$}] {};
  \node at (5.75,0) {$V'_1$};
  
  \draw (x) -- (v'1);
  \draw (x) -- (v'2);
  \draw (x) -- (v'T);
  
  \node[vertex] (v''1) at (3.25,-1) [label={[xshift=-7pt, yshift=-19pt] $v''_{1}$}] {};
  \node[vertex] (v''2) at (3.75,-1) {};
  \draw[fill] (4.15,-1) circle (0.3pt);
  \draw[fill] (4.25,-1) circle (0.3pt);
  \draw[fill] (4.35,-1) circle (0.3pt);
  \node[vertex] (v''T) at (4.75,-1) [label={[xshift=7pt, yshift=-19pt] $v''_{T}$}] {};
  \node at (5.75,-1) {$V''$};
  
  \node[vertex] (u'1) at (3.25,-2) [label={[xshift=-7pt, yshift=-19pt] $u'_{1}$}] {};
  \node[vertex] (u'2) at (3.75,-2) {};
  \draw[fill] (4.15,-2) circle (0.3pt);
  \draw[fill] (4.25,-2) circle (0.3pt);
  \draw[fill] (4.35,-2) circle (0.3pt);
  \node[vertex] (u'T) at (4.75,-2) [label={[xshift=7pt, yshift=-19pt] $u'_{T}$}] {};
  \node at (5.75,-2) {$U'_1$};
  
  \draw (v'1) -- (v''1) -- (u'1);
  \draw (v'2) -- (v''2) -- (u'2);
  \draw (v'T) -- (v''T) -- (u'T);
   
  \node[vertex] (u'T+1) at (3.25,-3.25) [label={[xshift=0pt, yshift=-20pt] $u'_{T+1}$}] {};
  \node[vertex] (u'T+2) at (3.75,-3.25) {};
  \draw[fill] (4.15,-3.25) circle (0.3pt);
  \draw[fill] (4.25,-3.25) circle (0.3pt);
  \draw[fill] (4.35,-3.25) circle (0.3pt);
  \node[vertex] (u'2T) at (4.75,-3.25) [label={[xshift=0pt, yshift=-20pt] $u'_{2T}$}] {};
  \node at (5.75,-3.25) {$U'_2$};
  
  \node[vertex] (u'2T+1) at (2.75,-4.25) [label={[xshift=-1pt, yshift=-20pt] $u'_{2T+1}$}] {};
  \node[vertex] (u'2T+2) at (3.25,-4.25) {};
  \node[vertex] (u'2T+3) at (3.75,-4.25) {};
  \draw[fill] (4.15,-4.25) circle (0.3pt);
  \draw[fill] (4.25,-4.25) circle (0.3pt);
  \draw[fill] (4.35,-4.25) circle (0.3pt);
  \node[vertex] (u'bT-2) at (4.75,-4.25) {};
  \node[vertex] (u'bT-1) at (5.25,-4.25) [label={[xshift=4pt, yshift=-20pt] $u'_{bT-1}$}] {};
  \node at (6,-4.25) {$U'_3$};
  
  \draw (u'T+1) -- (v'T+1);
  \draw (u'T+2) -- (v'T+2);
  \draw (u'2T) -- (v'2T);
  
  \draw[rounded corners, dashed] (1.5, -5.25) rectangle (6.75, 0.75) {};

  \draw (10,-2.25) ellipse (0.5 and 1.5);
  \node[vertex] (v'2T+1) at (10,-1.25) [label={[xshift=25pt, yshift=-12pt] $v'_{2T+1}$}] {};
  \draw[fill] (10,-2.15) circle (0.3pt);
  \draw[fill] (10,-2.25) circle (0.3pt);
  \draw[fill] (10,-2.35) circle (0.3pt);
  \node[vertex] (v'aT) at (10,-3.25) [label={[xshift=21pt, yshift=-12pt] $v'_{aT}$}] {};
  \node at (10,-0.25) {$V'_3$};
  
  \coordinate (E) at (6.85,-5.25);
  \coordinate (F) at (6.85,0.75);
  \coordinate (G) at (9.7,-3.8);
  \coordinate (H) at (9.7,-0.7);
  
  \draw (E) -- (G);
  \draw (F) -- (H);
  \draw ($(E)!0.05!(H)$) -- ($(E)!0.95!(H)$);
  \draw ($(F)!0.05!(G)$) -- ($(F)!0.95!(G)$);
 \end{tikzpicture}
 \caption{The graph~$H_t$, when $t \ge 1$.} \label{Fig:Ht_t>=1}
 \end{center}
\end{figure}

\begin{claim}
 For any rational number $t \ge 1$ the graph~$H_t$ has weighted toughness~$t$ with respect to the weight function~$w$ that assigns weight~1 to all the vertices of~$H_t$ except for the vertex~$x$, to which it assigns weight~$t$.
\end{claim}
\begin{proof}
 Let $S$ be an arbitrary cutset of~$H_t$. We need to show that $\omega(H_t - S) \le w(S) / t$.
 
 We can assume that either $V'_3 \cap S = \emptyset$ or $V'_3 \subseteq S$ since removing only a proper subset of~$V'_3$ does not disconnect anything from the graph. Similarly, we can also assume that either $U'' \cap S = \emptyset$ or $U'' \subseteq S$.
 
 \bigskip
 
 \textit{Case 1:} $V'_3 \cap S = \emptyset$ and $U'' \cap S = \emptyset$.
 
 Then $H_t - S$ has at most~2 components, and to obtain 2 components, the following must hold:
 \begin{enumerate}
  \item[--] $u'_{T+l} \in S$ or $v'_{T+l} \in S$ for all $l \in [T]$, and
  \item[--] $x \in S$ or $V'_1 \subseteq S$.
 \end{enumerate}
 
 Hence $w(S) \ge T + t$ and
  \[ \omega(H_t - S) = 2 \le \frac{T+t}{t} \le \frac{w(S)}{t} \text{.} \]
 
 \bigskip
 
 \textit{Case 2:} $V'_3 \cap S = \emptyset$ and $U'' \subseteq S$.
 
 Now we can assume that $x \notin S$ since after the removal of~$U''$ removing~$x$ does not disconnect anything from the graph. Similarly, we can also assume that $V'_2 \nsubseteq S$. Then $H_t - S$ has at most~3 components. To obtain three components, the following must hold:
 \begin{enumerate}
  \item[(i)] $u'_{T+l} \in S$ or $v'_{T+l} \in S$ for all $l \in [T]$ (but $V'_2 \nsubseteq S$), and
  \item[(ii)] $V'_1 \subseteq S$.
 \end{enumerate}
 Hence $w(S) \ge  T' + 2T = \lceil 2t \rceil + T$ and
  \[ \omega(H_t - S) = 3 \le \frac{\lceil 2t \rceil + T}{t} = \frac{w(S)}{t} \text{.} \]
 To obtain two components, either (i) or (ii) must hold; in both cases $w(S) \ge \lceil 2t \rceil$ and
  \[ \omega(H_t - S) = 2 \le \frac{\lceil 2t \rceil}{t} \le \frac{w(S)}{t} \text{.} \] 
 
 \bigskip
 
 \textit{Case 3:} $V'_3 \subseteq S$.
 
 First we show that the following assumptions can be made for~$S$.
 
 \begin{description}
  \item[] (1) $(U'_1 \cup U'_2 \cup U'_3) \cap S = \emptyset$.
 
   After the removal of~$V'_3$, removing any of the vertices of $U'_1 \cup U'_2 \cup U'_3$ does not disconnect anything from the graph.
 
  \item[] (2) There exists at most one $l \in [T]$ for which $v'_{T+l} \notin S$, i.e. $|V'_2 \setminus S| \le 1$.
  
   Suppose that there exist $l_1, l_2 \in [T]$ for which $l_1 \ne l_2$ and $v'_{T+l_1}, v'_{T+l_2} \notin S$. By assumption~(1), considering the cutset $S' = S \cup \{ v'_{T+l_2} \}$ instead of~$S$ increases both the number of components and the weight of the removed vertex set by~1. Hence it is enough to show that 
    \[ \omega(H_t - S') \le \frac{w(S')}{t} \]
   since it implies
    \[ \omega(H_t - S) = \omega(H_t - S') - 1 \le \frac{w(S')}{t} - 1 = \frac{w(S)+1}{t} - 1 \le \frac{w(S)}{t} \text{,} \]
   where the last inequality is valid since $t \ge 1$.
 
  \item[] (3) For all $l \in [T]$ if $v'_{l} \in S$, then $v''_l \notin S$.
 
   After the removal of $V'_3$ and $v'_l$ removing $v''_l$ does not disconnect anything from the graph.
 
  \item[] (4) For all $l \in [T]$ if $v'_{l} \notin S$, then $v''_l \in S$.
 
   Suppose that there exists $l \in [T]$ for which $v'_{l}, v''_{l} \notin S$. By assumption~(1), considering the cutset $S' = S \cup \{ v''_{l} \}$ instead of~$S$ increases both the number of components and the weight of the removed vertex set by~1. Hence, similarly as in assumption~(2), it is enough to show that 
    \[ \omega(H_t - S') \le \frac{w(S')}{t} \text{.} \]
 
  \item[] (5) $\big| (V'_1 \cup V'') \cap S \big| = T$.
 
   It follows directly from assumptions~(3) and~(4).
 \end{description}
 
 \bigskip
 
 \textit{Case 3.1:} ($V'_3 \subseteq S$ and) $U'' \subseteq S$.
 
 Now we can assume that $x \notin S$ since after the removal of~$U''$ removing~$x$ does not disconnect anything from the graph. Similarly, by assumption~(2), we can also assume that $V'_2 \nsubseteq S$, i.e. $|V'_2 \cap S| = T-1$. Hence
  \[ w(S) = |V_3'| + |U''| + \big| (V'_1 \cup V'') \cap S \big| + |V'_2 \cap S| = (aT-2T) + T' + T + (T-1) = aT + T' -1 \]
 and every component of $H_t - S$ contains exactly one of the vertices $u'_1, \ldots, u'_{bT-1}, x$, i.e.
  \[ \omega(H_t - S) = bT = \frac{aT}{t} \le \frac{aT + T' - 1}{t} = \frac{w(S)}{t} \text{.} \]
 
 \bigskip
 
 \textit{Case 3.2:} ($V'_3 \subseteq S$ and) $U'' \cap S = \emptyset$.
 
 In this case we can make some further assumptions for $S$.
 
 \begin{description}
  \item[] (6) If $V'_1 \subseteq S$, then $x \notin S$.
 
   After the removal of~$V'_1$ removing~$x$ does not disconnect anything from the graph.
 
  \item[] (7) If $V'_1 \nsubseteq S$, then $x \in S$.
 
   Suppose that $x \notin S$. Then considering the cutset $S' = S \cup \{ x \}$ instead of~$S$ increases the number of components by~1 and the weight of the removed vertex set by~$t$. Hence it is enough to show that $\omega(H_t - S') \le w(S')/t$ since it implies
    \[ \omega(H_t - S) = \omega(H_t - S') - 1 \le \frac{w(S')}{t} - 1 = \frac{w(S)+t}{t} - 1 = \frac{w(S)}{t} \text{.} \]
 
  \item[] (8) $V'_2 \subseteq S$.
 
   Suppose that $V'_2 \nsubseteq S$. Then by assumption~(2), there exists $l \in [T]$ for which $V'_2 \setminus S = \{ v'_{T+l} \}$. But by assumption~(1), considering the cutset $S' = S \cup \{ v'_{T+l} \}$ instead of~$S$ increases both the number of components and the weight of the removed vertex set by~1. Then, similarly as in assumption~(2), it is enough to show that $\omega(H_t - S') \le w(S')/t$.
 \end{description}
 
 \bigskip
 
 \textit{Case 3.2.1:} ($V'_3 \subseteq S$, $U'' \cap S = \emptyset$ and) $V'_1 \subseteq S$.
 
 Hence
  \[ w(S) = |V'_3| + |V'_2| + |V'_1| = aT \]
 and
  \[ \omega(H_t - S) = bT = \frac{w(S)}{t} \text{.} \]
 
 \bigskip
 
 \textit{Case 3.2.2:} ($V'_3 \subseteq S$, $U'' \cap S = \emptyset$ and) $V'_1 \nsubseteq S$.
 
 Hence
  \[ w(S) = |V'_3| + |V'_2| + |(V'_1 \cup V'') \cap S| + w(x) = aT + t \]
 and
  \[ \omega(H_t - S) = bT + 1 = \frac{w(S)}{t} \text{.} \]
 
 \bigskip
 
 Therefore $H_t$ is weighted \mbox{$t$-tough} with respect to~$w$ (meaning that the weighted toughness of~$H_t$ is at least~$t$).
 
 Consider the cutset
  \[ S = V'_1 \cup V'_2 \cup V'_3 \text{.} \]
 Since $w(S) = aT$ and 
  \[ \omega(H_t - S) = bT = \frac{w(S)}{t} \text{,} \]
 the weighted toughness of~$H_t$ with respect to~$w$ is at most~$t$.
 
 Thus the weighted toughness of~$H_t$ with respect to~$w$ is exactly~$t$.
\end{proof}

Deleting an edge may decrease the weighted toughness, and now we delete edges not induced by~$U''$ until the weighted toughness with respect to the weight function~$w$ remains at least~$t$ but the deletion of any other edge not induced by~$U''$ would result in a graph with weighted toughness less than~$t$. Let~$H'_t$ denote the obtained graph.

According to the following claim we could not delete the edges induced by~$V'_1$ or incident to any of the vertices of $\{ x \} \cup V'_2 \cup U''$.

\begin{claim}
 Let $t \ge 1$ be a rational number. For any edge $e \in E(H_t)$ induced by $V'_1$ or incident to any of the vertices of $\{ x \} \cup V'_2 \cup U''$, there exists a cutset $S = S(e) \subseteq V(H_t)$ in $H_t-e$ for which 
  \[ \omega \big( (H_t-e) - S \big) > \frac{w(S)}{t} \text{.} \]
\end{claim}
\begin{proof}
 Let $e \in E(H_t)$ be an arbitrary edge induced by~$V'_1$ or incident to any of the vertices of $\{ x \} \cup V'_2 \cup U''$.
 
 \bigskip
 
 \textit{Case 1:} $e$ is incident to a vertex of $\{ x \} \cup V'_2$.
 
 Let $y \in \{ x \} \cup V'_2$ denote one of the endpoints of~$e$, and let~$z$ denote the other one. Let $S$ be the neighborhood of the vertex~$y$ except for~$z$. Since $y$ has degree $\lceil 2t \rceil$ and all of its neighbors have weight~1,
  \[ w(S) = \lceil 2t \rceil - 1 \text{.} \]
 Since the removal of~$S$ from $H_t - e$ leaves the vertex~$y$ isolated,
  \[ \omega \big( (H_t - e) - S \big) \ge 2 = \frac{2t}{t} > \frac{\lceil 2t \rceil - 1}{t} = \frac{w(S)}{t} \text{.} \]
 
 \bigskip
 
 \textit{Case 2:} $e$ is incident to a vertex of~$U''$.
 
 If $e$ is incident to a vertex of~$U''$, then either it is incident to a vertex of $\{ x \} \cup V'_2$ and this case was already settled in Case~1, or it is induced by~$U''$ and therefore it was not allowed to be deleted.
 
 \bigskip
 
 \textit{Case 3:} $e$ is induced by~$V'_1$, i.e. $e = v'_{l_1} v'_{l_2}$ for some $l_1, l_2 \in [T]$, $l_1 \ne l_2$.
 
 Then
  \[ S = \big( V'_1 \setminus \{ v'_{l_1}, v'_{l_2} \} \big) \cup \{ v''_{l_1}, v''_{l_2} \} \cup V'_2 \cup V'_3 \cup \{ x \} \]
 is a cutset in $H_t - e$ such that
  \[ w(S) = (T-2) + 2 + T + (aT-2T) + t = aT + t \]
 and
  \[ \omega \big( (H_t - e) - S \big) = bT + 2 = \frac{aT + t}{t} + 1 = \frac{w(S)}{t} + 1 > \frac{w(S)}{t} \text{.} \]
\end{proof}

\begin{claim} \label{properties_of_Ht_t>=1}
 Let $t \ge 1$ be a rational number and $H''_t = H'_t - \{ x \}$. Then the following hold.
 \begin{enumerate}
  \item[(i)] The graph~$H''_t$ is connected.
  \item[(ii)] For any cutset~$S$ of~$H''_t$,
   \[ \omega(H''_t - S) \le \frac{|S|}{t} + 1 \text{.} \]
  \item[(iii)] If $V'_1 \subseteq S$ holds for a cutset~$S$ of~$H''_t$, then
   \[ \omega(H''_t - S) \le \frac{|S|}{t} \text{.} \]
  \item[(iv)] For any edge $e \in E(H''_t)$ not induced by~$U''$ there exists a vertex set $S=S(e)$ whose removal from $H''_t - e$ disconnects the graph and
   \[ \omega \big( (H''_t - e) - S \big) > \frac{|S|}{t} \text{.} \]
 \end{enumerate}
\end{claim}
\begin{proof} $\empty$
 \begin{enumerate}
  \item[(i)] Suppose to the contrary that $H''_t$ is not connected. Then $x$ is a cut-vertex in~$H'_t$. Since the weighted toughness of~$H'_t$ with respect to~$w$ is~$t$,
   \[ 2 \le \omega \big( H'_t - \{ x \} \big) \le \frac{w(x)}{t} = \frac{t}{t} = 1 \text{,} \]
  which is a contradiction.
  
  \bigskip
  
  \item[(ii)] Let $S$ be an arbitrary cutset of~$H''_t$. Since $S$ is a cutset in~$H''_t$, the vertex set $S \cup \{ x \}$ is a cutset in~$H'_t$, and
   \[ \omega(H''_t - S) = \omega \big( H'_t - (S \cup \{ x \}) \big) \le \frac{w(S \cup \{ x \})}{t} = \frac{|S|+t}{t} = \frac{|S|}{t} + 1 \text{.} \]
  
  \bigskip
  
  \item[(iii)] Let $S$ be a cutset of~$H''_t$ for which $V'_1 \subseteq S$. We can assume that $U'' \cap S = \emptyset$ since removing any of the vertices of~$U''$ from~$H''_t$ does not disconnect anything from the graph. Then all the neighbors of the vertex~$x$ belong to the same component in $H''_t - S$, so $S$ is a cutset in~$H'_t$ as well and
   \[ \omega(H''_t - S) = \omega(H'_t - S) \le \frac{w(S)}{t} = \frac{|S|}{t} \text{,} \]
  where the last equality is valid since $x \notin S$.
 
  \bigskip
 
  \item[(iv)] Let $e \in E(H''_t)$ be an arbitrary edge not induced by $U''$. Then by the properties of $H'_t$, there exists a vertex set $S \subseteq V(H'_t)$ whose removal from $H'_t - e$ disconnects the graph and
   \[ \omega \big( (H'_t-e) - S \big) > \frac{w(S)}{t} \ge \frac{|S|}{t} \text{,} \]
  where the last inequality is valid since $t \ge 1$. Let $S' = S \setminus \{ x \}$. Then
   \[ \omega \big( (H''_t-e) - S' \big) \ge \omega \big( (H'_t-e) - S \big) > \frac{|S|}{t} \ge \frac{|S'|}{t} \text{.} \]
 \end{enumerate}
\end{proof}

\subsection{The cutsets $X$ and $Y_1, \ldots, Y_{T}$ in $H''_t$ when $t \ge 1$} \label{subsection_XY}

Let 
 \[ X = V'_1 \cup V'_2 \cup V'_3 \]
and for all $l \in [T]$ let
 \[ Y_l = \big( V'_1 \setminus \{ v'_{l} \} \big) \cup \{ v''_l \} \cup V'_2 \cup V'_3 \text{.} \]

\begin{prop} \label{proposition_XY}
 The sets~$X$ and $Y_1, \ldots, Y_{T}$ are all cutsets in~$H''_t$ and 
  \[ \omega(H''_t - X) = \frac{|X|}{t} \text{,} \]
 and
  \[ \omega(H''_t - Y_l) = \frac{|Y_l|}{t} + 1\]
 for all $l \in [T]$.
\end{prop}
\begin{proof}
 It is easy to see that
 \[ \omega(H''_t - X) = bT = \frac{aT}{t} = \frac{|X|}{t} \]
 and
 \[ \omega(H''_t - Y_l) = bT + 1 = \frac{aT}{t} + 1 = \frac{|Y_l|}{t} + 1 \text{.} \]
\end{proof}

\subsection{The proof of Theorem~\ref{minttough_dp_complete} when $t \ge 1$}

\begin{thm} \label{minttough_dp_complete_t>=1}
 For any rational number $t \ge 1$, the problem {\scshape Min-\mbox{$t$-Tough}} is DP-complete.
\end{thm}
\begin{proof}
 Let $t \ge 1$ be a rational number. In Proposition~\ref{min_t_tough_DP} we already proved that the problem {\scshape Min-\mbox{$t$-Tough}} is in DP. To show that it is DP-hard, we reduce the variant of {\scshape \mbox{$\alpha$-Critical}} defined in Proposition~\ref{alphacrit_variant} to it.
 
 Let $T = \lceil t \rceil$, and $T'= \lceil 2t \rceil - \lceil t \rceil$, and $M = \big\lceil 2 \lceil t \rceil/\lceil 2t \rceil \big\rceil$ as before. Let $a,b$ be the smallest positive integers such that $b \ge 3$ and $t = a/b$, let $G$ be an arbitrary \mbox{$3$-}con\-nected graph on the vertices $v_1, \ldots, v_n$ with $n \ge t+1$, let $k$ be a positive integer that is divisible by~$a$ and let $G_{t,k}$ be defined as follows. For all $i \in [n], j \in [k], m \in [M]$ let
  \[ V_{i,j,m} = \big\{ v_{i,j,l,m} \bigm| l \in [T] \big\} \text{.} \]
 For all $i \in [n]$ let
  \[ V_i = \bigcup_{\substack{j \in [k], \\ m \in [M]}} V_{i,j,m} \]
 and place a clique on the vertices of $V_i$. For all $i_1, i_2 \in [n]$ if $v_{i_1} v_{i_2} \in E(G)$, then place a complete bipartite graph on $(V_{i_1}; V_{i_2})$. (This subgraph is denoted by $\tilde{G}$ in Figure~\ref{minttough_deterministic_constr_t>=1}.) For all $i \in [n], j \in [k], m \in [M]$ ``glue'' the graph~$H_t''$ to the vertex set $V_{i,j,m}$ by identifying $v_{i,j,l,m}$ with the vertex $v'_l$ of $H''_t$ for all $l \in [T]$. For all $i \in [n], j \in [k], m \in [M]$ let $H^{i,j,m}$, $U''_{i,j,m}$ and $X_{i,j,m}$ denote the $(i,j,m)$-th copies of $H_t''$, $U''$ and $X$, respectively. For all $i \in [n], j \in [k], l \in [T'], m \in [M]$ let $u''_{i,j,l,m}$ denote the $(i,j,m)$-th copy of $u''_l$, and for all $i \in [n], j \in [k], l \in [T], m \in [M]$ let $Y_{i,j,l,m}$ denote the $(i,j,m)$-th copy of $Y_l$. For all $j \in [k], m \in [M]$ add the vertex set
  \[ W_{j,m} = \big\{ w_{j,l,m} \bigm| l \in [T'] \big\} \]
 to the graph and for all $i \in [n], j \in [k], l \in [T'], m \in [M]$ connect $w_{j,l,m}$ to $u''_{i,j,l,m}$. Let
  \[ V = \bigcup_{i \in [n]} V_i \text{,} \qquad U'' = \bigcup_{\substack{i \in [n], \\ j \in [k], \\ m \in [M]}} U''_{i,j,m} \text{,} \qquad W = \bigcup_{\substack{j \in [k], \\ m \in [M]}} W_{j,m} \text{,} \]
 and let
  \[ U = \left( \bigcup_{\substack{i \in [n], \\ j \in [k], \\ m \in [M]}} V(H^{i,j,m}) \right) \setminus V \text{.} \]
 Add the vertex set
  \[ W' = \{ w'_1, \ldots, w'_{(MT'/t-1) k} \} \]
 to the graph and place the bipartite graph $H^{**}_{t,k}$ on $(W;W')$. See Figure~\ref{minttough_deterministic_constr_t>=1}. Now~$k$ is part of the input of the problem {\scshape \mbox{$\alpha$-Critical}}, therefore the graph $H^{**}_{t,k}$ must be constructed in polynomial time and by Theorem~\ref{split_general_thm}, this can be done. On the other hand, $t$~is not part of the input of the problem {\scshape Min-\mbox{$t$-Tough}}, therefore the graph~$H''_t$ can be constructed in advance. Hence, $G_{t,k}$ can be constructed from~$G$ in polynomial time.
 
 \begin{figure}[H]
 \begin{center}
 \begin{tikzpicture}[scale=1.25]
  \tikzstyle{vertex}=[draw,circle,fill=black,minimum size=3,inner sep=0]
   
  \draw (-0.1,0) ellipse (1.1 and 2.25);
  \draw (0,1.45) ellipse (0.5 and 0.5);
  \draw[fill] (0,0.1) circle (0.3pt);
  \draw[fill] (0,0) circle (0.3pt);
  \draw[fill] (0,-0.1) circle (0.3pt);
  \draw (0,-1.45) ellipse (0.5 and 0.5);
  
  \node at (-1,2.5) {$\tilde{G}$};
  \node at (0,1.45) {$V_{1,1,1}$};
  \node at (0,-1.45) {$V_{n,k,M}$};
  
  \draw (1,1.45) ellipse (1.6 and 0.75);
  \draw (1,-1.45) ellipse (1.6 and 0.75);
  
  \node at (1.05,1.75) {$H^{1,1,1}$};
  \node at (1.05,-1.75) {$H^{n,k,M}$};
  
  \draw (2,1.45) ellipse (0.5 and 0.5);
  \draw[fill] (0,0.1) circle (0.3pt);
  \draw[fill] (0,0) circle (0.3pt);
  \draw[fill] (0,-0.1) circle (0.3pt);
  \draw (2,-1.45) ellipse (0.5 and 0.5);
  
  \node at (2,1.45) {$U''_{1,1,1}$};
  \node at (2,-1.45) {$U''_{n,k,M}$};
  
  \node[vertex] (u1) at (2,1.8) [label={[xshift=20pt, yshift=0pt] $u''_{1,1,1,1}$}] {};
  \node[vertex] (u2) at (2,1.1) {};
  \node[vertex] (u3) at (2,-1.1) {};
  \node[vertex] (u4) at (2,-1.8) [label={[xshift=25pt, yshift=-22pt] $u''_{n,k,T',M}$}] {};
  
  \node[vertex] (w1) at (3.5,1) [label={[xshift=15pt, yshift=-7pt] $w_{1,1,1}$}] {};
  \draw[fill] (3.5,0.1) circle (0.3pt);
  \draw[fill] (3.5,0) circle (0.3pt);
  \draw[fill] (3.5,-0.1) circle (0.3pt);
  \node[vertex] (w2) at (3.5,-1) [label={[xshift=17pt, yshift=-15pt] $w_{k,T',M}$}] {};
  
  \node at (3.5,1.75) {$W$};
  
  \draw (u1) -- (w1);
  \draw (u2) -- (w2);
  \draw (u3) -- (w1);
  \draw (u4) -- (w2);
  
  \node[vertex] (w'1) at (4.5,1) [label={[xshift=15pt, yshift=-7pt] $w'_1$}] {};
  \draw[fill] (4.5,0.1) circle (0.3pt);
  \draw[fill] (4.5,0) circle (0.3pt);
  \draw[fill] (4.5,-0.1) circle (0.3pt);
  \node[vertex] (w'b) at (4.5,-1) [label={[xshift=32pt, yshift=-20pt] $w'_{(MT'/t - 1)k}$}] {};
  
  \node at (4.5,1.75) {$W'$};
  
  \draw (4,0) ellipse (1 and 1.5);
  
  \node at (4,0) {$H_{t,k}^{**}$};
 \end{tikzpicture}
 \caption{The graph $G_{t,k}$, when $t \ge 1$.} \label{minttough_deterministic_constr_t>=1}
 \end{center}
 \end{figure}
 
 To show that $G$ is \mbox{$\alpha$-critical} with $\alpha(G) = k$ if and only if $G_{t,k}$ is minimally \mbox{$t$-tough}, first we prove the following lemma.
 
 \begin{lemma} \label{dp_complete_main_lemma_t>=1}
  Let $G$ be an arbitrary \mbox{$3$-}con\-nected graph on $n \ge t+1$ vertices with $\alpha(G) \le k$. Then $G_{t,k}$ is \mbox{$t$-tough}.
 \end{lemma}
 \begin{proof}
  Let $S \subseteq V(G_{t,k})$ be a cutset in $G_{t,k}$. We need to show that $\omega(G_{t,k} - S) \le |S|/t$.
  
  \bigskip
  
  \textit{Case 1:} $W \subseteq S$.
  
  After the removal of~$W$, the vertices of~$W'$ are isolated, therefore we can assume that $W' \cap S = \emptyset$.
  
  Let
   \[ C = C(S) = \big\{ (i,j,m) \in [n] \times [k] \times [M] \bigm| V_{i,j,m} \subseteq S \big\} \text{,} \]
   \[ c_{i,j,m} = \big| V(H^{i,j,m}) \cap S \big| - T\]
  for all $(i,j,m) \in C$, and 
   \[ d_{i,j,m} = \big| V(H^{i,j,m}) \cap S \big| \]
  for all $(i,j,m) \in \big( [n] \times [k] \times [M] \big) \setminus C$. Let
   \[ D = D(S) = \left\{ (i,j,m) \in \big( [n] \times [k] \times [M] \big) \setminus C ~ \middle| ~ d_{i,j,m} > 0 \right\} \text{.} \]
  
  Using these notations it is clear that
   \[ |S| = |C| \cdot T + \sum_{(i,j,m) \in C} c_{i,j,m} + \sum_{(i,j,m) \in D} d_{i,j,m} + kMT' \text{.} \]
  By the assumption that $W \subseteq S$, in $G_{t, k}-S$ the $(MT'/t-1)k$ vertices of $W'$ are isolated. Since $\alpha \big( G_{t, k}[V] \big) = \alpha(G)$, the removal of $V \cap S$ from $G_{t, k}[V]$ leaves at most $\alpha(G)$ components. By Claim~\ref{properties_of_Ht_t>=1}, for any $(i,j,m) \in C$ the removal of $V(H^{i,j,m}) \cap S$ from $H^{i,j,m}$ leaves at most 
  \[ \max \left( \frac{|V(H^{i,j,m}) \cap S|}{t}, 1 \right) = \max \left( \frac{c_{i,j,m}+T}{t}, 1 \right) = \frac{c_{i,j,m}+T}{t} \]
  components. By Claim~\ref{properties_of_Ht_t>=1}, for any $(i,j,m) \in D$ the removal of $V(H^{i,j,m}) \cap S$ from $H^{i,j,m}$ leaves at most
  \[ \max \left( \frac{|V(H^{i,j,m}) \cap S|}{t}+1, 1 \right) = \max \left( \frac{d_{i,j,m}}{t}+1, 1 \right) = \frac{d_{i,j,m}}{t}+1 \]
  components, but the component of the remaining vertices of $V_{i,j,m}$ has been already counted. Hence
  \begin{gather*}
   \omega(G_{t, k}-S) \le \left( \frac{MT'}{t}-1 \right) k + \alpha(G) + \sum_{(i,j,m) \in C} \frac{c_{i,j,m} + T}{t} + \sum_{(i,j,m) \in D} \frac{d_{i,j,m}}{t} \\
   \le \frac{kMT' + |C| \cdot T + \sum_{(i,j,m) \in C} c_{i,j,m} + \sum_{(i,j,m) \in D} d_{i,j,m} }{t} = \frac{|S|}{t} \text{.}
  \end{gather*}
  
  \bigskip
  
  \textit{Case 2:} $W \nsubseteq S$.
  
  There are four types of components in $G_{t,k}-S$:
  \begin{enumerate}
   \item[(a)] components containing at least one vertex of~$V$,
   \item[(b)] components containing at least one vertex of~$U$ but no vertices of~$V$,
   \item[(c)] components containing at least one vertex of~$W$ but no vertices of $U \cup V$,
   \item[(d)] isolated vertices of~$W'$.
  \end{enumerate}
  
  Let $w_{j_0,l_0,m_0} \in W \setminus S$ be fixed. First we show that the following assumptions can be made for~$S$.
  
  \begin{description}
   \item[] (1) $S \cap U'' = \emptyset$.
  
    Obviously, the number of vertices of~$W$ that belong to a component of type~(c) is at most $|S \cap U''|/n$. Since the neighborhood of any vertex of~$U''$ spans a clique in $G_{t,k}$, considering the cutset
     \[ S' = (S \setminus U'') \cup \big\{ w \in W \bigm| \text{$w$ belongs to a component of type~(c)} \big\} \]
    instead of~$S$ can only increase the number of components of types~(a), (b) and~(d), while it decreases the number of components of type~(c) to 0, i.e., by at most $|S \cap U''|/n$. Hence,
     \[ |S'| \le |S| - |S \cap U''| + \frac{|S \cap U''|}{n} \]
    and
     \[ \omega(G_{t,k} - S') \ge \omega(G_{t,k} - S) - \frac{|S \cap U''|}{n} \text{.} \]
    Then it is enough to prove that $\omega(G_{t,k}-S') \le |S'|/t$ since it implies
    \begin{gather*}
     \omega(G_{t,k} - S) \le \omega(G_{t,k} - S') + \frac{|S \cap U''|}{n} \le \frac{|S'|}{t} + \frac{|S \cap U''|}{n} \\ 
     \le \frac{|S| - |S \cap U''| + |S \cap U''|/n}{t} + \frac{|S \cap U''|}{n} = \frac{|S|}{t} - |S \cap U''| \cdot \frac{n-t-1}{nt} \le \frac{|S|}{t} \text{,}
    \end{gather*}
    where the last inequality is valid since $n \ge t+1$.
    
   \item[] (2) There are no components of type~(c) in $G_{t,k}-S$.
   
    It follows directly from assumption~(1).
  
   \item[] (3) $\big| \{ i \in [n] \mid V_{i,j_0,m_0} \subseteq S \} \big| \le \lceil T'/t \rceil$.
  
    Let 
     \[ I = I(w_{j_0, l_0,m_0}, S) = \big\{ i \in [n] \bigm| V_{i,j_0, m_0} \subseteq S \big\} \]
    and suppose that $|I| \ge \lceil T'/t \rceil + 1$. By assumption~(1), the component of $w_{j_0,l_0,m_0}$ contains every vertex of $\bigcup_{i=1}^n U_{i,j_0,m_0}$ and therefore all the remaining vertices of $W_{j_0,m_0}$. Now considering the cutset
     \[ S' = S \cup \{ W_{j_0,m_0} \} \]
    instead of~$S$ increases the number of removed vertices by at most~$T'$, and it increases the number of components by at least $\lceil T'/t \rceil$ since it disconnects the vertex sets $U_{i,j_0,m_0}$, $i \in I$ from each other. Then it is enough to show that $\omega(G_{t,k}-S') \le |S'|/t$ since it implies
     \[ \omega(G_{t,k} - S) \le \omega(G_{t,k} - S') - \left\lceil \frac{T'}{t} \right\rceil \le \frac{|S'|}{t} - \left\lceil \frac{T'}{t} \right\rceil \le \frac{|S|+T'}{t} - \left\lceil \frac{T'}{t} \right\rceil \le \frac{|S|}{t} \text{.} \]
    Proceeding further, we can obtain a cutset $S^*$ for which $W \subseteq S^*$ holds; and such sets were already handled in Case 1.
  
   \item[] (4) $(G_{t,k}-S)[V]$ is connected, i.e.~there is only one component of type~(a).
  
    Since $t \ge 1$,
     \[ \left\lceil \frac{T'}{t} \right\rceil \le \left\lceil \frac{t+1}{t} \right\rceil = 1 + \left\lceil \frac{1}{t} \right\rceil \le 2 \text{.} \]
    Since $G$ is \mbox{3-}con\-nected, assumption~(2) implies that $(G_{t,k}-S)[V]$ is connected.
  \end{description}
  
  Using the previous notations,
   \[ |S| = |C| \cdot T + \sum_{(i,j,m) \in C} c_{i,j,m} + \sum_{(i,j,m) \in D} d_{i,j,m} + |S \cap (W \cup W')| \text{.} \]
  
  By assumption~(2), there are no components of type~(c), and by assumption~(4), there is only one component of type~(a). By the properties of $H^{**}_{t,k}$, the removal of $S \cap (W \cup W')$ from $H^{*}_{t,k}$ leaves at most
   \[ \max \left( \frac{|S \cap (W \cup W')|}{t}, 1 \right) \]
  components, one of them is the component of $w_{j_0,l_0,m_0}$, hence there are at most
   \[ \max \left( \frac{|S \cap (W \cup W')|}{t}, 1 \right) - 1 \le \frac{|S \cap (W \cup W')|}{t} \]
  components of type~(d). Similarly as before, for any $(i,j,m) \in C$ the removal of $V(H^{i,j,m}) \cap S$ from $H^{i,j,m}$ leaves at most 
   \[ \frac{c_{i,j,m}+T}{t} \]
  components, all of which can be of type~(b). For any $(i,j,m) \in D$ the removal of $V(H^{i,j,m}) \cap S$ from $H^{i,j,m}$ leaves at most
   \[ \frac{d_{i,j,m}}{t}+1 \]
  components; the component of the remaining vertices of $V_{i,j,m}$ is of type~(a), all the others can be of type~(b). By assumption~(1), all the vertices of $\bigcup_{i=1}^n U_{i,j_0,m_0}$ belong to the component of $w_{j_0,l_0,m_0}$, hence the component of $w_{j_0,l_0,m_0}$ has been counted multiple times (more than once). Therefore,  
   \[ \omega(G_{t,k}-S) \le \left( 1 + \sum_{(i,j,m) \in C} \frac{c_{i,j,m} + T}{t} + \sum_{(i,j,m) \in D} \frac{d_{i,j,m}}{t} + \frac{|S \cap (W \cup W')|}{t} \right) - 1 = \frac{|S|}{t} \text{.} \]
  
  \bigskip
  
  Thus, $\tau(G_{t, k}) \ge t$.
 \end{proof}
 
 Now we return to the proof of Theorem~\ref{minttough_dp_complete_t>=1} and we show that $G$ is \mbox{$\alpha$-critical} with $\alpha(G) = k$ if and only if $G_{t, k}$ is minimally \mbox{$t$-tough}.
 
 Let us assume that $G$ is \mbox{$\alpha$-critical} with $\alpha(G) = k$. By Lemma~\ref{dp_complete_main_lemma_t>=1}, the graph $G_{t, k}$ is \mbox{$t$-tough}, i.e.~$\tau(G_{t,k}) \ge t$.
 
 Let $I$ be an independent vertex set of size $\alpha(G)$ in~$G$, and recall the definition of the sets~$X$ and $Y_1, \ldots, Y_T$ constructed in Subsection~\ref{subsection_XY}. Let
  \[ J = \big\{ i \in [n] \bigm| v_i \in I \big\} \]
 and
  \[ S = \left( \bigcup_{i \in J} Y_{i,1,1,1} \right) \cup \left( \bigcup_{\substack{i \in J, \\ j \in [k] \setminus \{1\}, \\ m \in [M] \setminus \{1\}}} X_{i,j,m} \right) \cup \left( \bigcup_{\substack{i \notin J, \\ j \in [k], \\ m \in [M]}} X_{i,j,m} \right) \cup W \text{.} \]
 Then $S$ is a cutset in $G_{t,k}$ with
  \[|S| = nkMaT + kMT' \]
 and after the removal of~$S$ from $G_{t,k}$, the vertices of~$W'$ are isolated and the other components of $G_{t,k}-S$ are exactly the components of $H^{i,j,m} - \big( S \cap V(H^{i,j,m}) \big)$ for all $(i,j,m) \in [n] \times [k] \times [M]$. By Proposition~\ref{proposition_XY},
  \[ \omega(H^{i,j,m} - X_{i,j,m}) = bT \]
 and
  \[ \omega(H^{i,j,m} - Y_{i,1,1,1}) = bT+1 \]
 for all $(i,j,m) \in [n] \times [k] \times [M]$. Since $|J| = \alpha(G)$ and
  \[ |W'| = \left( \frac{MT'}{t} - 1 \right)k \text{,} \]
 it follows that
  \[ \omega (G_{t, k} - S) = nkMbT + \alpha(G) + \left( \frac{MT'}{t}-1 \right)k = nkMbT + \frac{kMT'}{t} = \frac{nkMaT + kMT'}{t} = \frac{|S|}{t} \text{,} \]
 so $\tau(G_{t, k}) \le t$.
 
 Therefore, $\tau(G_{t,k}) = t$.
 
 \bigskip
 
 Let $e \in E(G_{t, k})$ be an arbitrary edge. We need to show that $\tau(G_{t,k} - e) < t$. Now we have four cases.
 
 \bigskip
 
 \textit{Case 1:} $e$ has an endpoint in~$U''$.
 
 Then this endpoint has degree $\lceil 2t \rceil - 1$ in $G_{t,k}-e$, so 
 \[ \tau(G_{t, k} - e) \le \frac{\lceil 2t \rceil - 1}{2} < \frac{2t}{2} = t \text{.} \]
 
 \bigskip
 
 \textit{Case 2:} $e$ has an endpoint in~$W'$.
 
 By the properties of $H^*_{t,k}$, there exists a cutset $S \subseteq W$ in $H^*_{t,k}-e$ for which 
  \[ \omega \big( (H^*_{t,k} - e) - S \big) > \frac{|S|}{t} \text{.} \]
 Note that $S$ is also a cutset in $G_{t, k}-e$ and 
  \[ \omega \big( (G_{t, k} - e) - S \big) = \omega \big( (H^*_{t,k} - e) - S \big) > \frac{|S|}{t} \text{,} \]
 so $\tau(G_{t, k} - e) < t$.

 \bigskip
 
 \textit{Case 3:} $e$ is induced by $H^{i_0,j_0,m_0}$ for some $i_0 \in [n], j_0 \in [k], m_0 \in [M]$.
 
 The case when $e$ is induced by $U''_{i_0, j_0, m_0}$ was already covered in Case~1. So assume that $e$~is not induced by $U''_{i_0,j_0,m_0}$. Then by Claim~\ref{properties_of_Ht_t>=1}, there exists a vertex set $S \subseteq V(H''_t)$ for which 
  \[ \omega \big( (H''_t - e) - S \big) > \frac{|S|}{t} \text{.} \]
 Consider the $(i_0,j_0,m_0)$-th copy of the vertex set~$S$ in $G_{t,k}-e$; let us denote it with $S_{i_0, j_0,m_0}$. If $V_{i_0,j_0,m_0} \subseteq S_{i_0, j_0, m_0}$, then $S_{i_0, j_0, m_0}$ is a cutset in $G_{t,k}-e$ and
  \[ \omega \big( (G_{t,k} - e) - S_{i_0, j_0, m_0} \big) = \omega \big( (H''_t - e) - S \big) > \frac{|S|}{t} \text{,} \]
 so $\tau(G_{t, k} - e) < t$. Assume that $V_{i_0,j_0,m_0} \nsubseteq S_{i_0, j_0, m_0}$. Let $I$ be an independent vertex set of size $\alpha(G)$ in~$G$ that contains~$v_{i_0}$ (by Proposition~\ref{alphacrit_max_indep}, such an independent vertex set exists).
 Let
  \[ J = \big\{ i \in [n] \bigm| v_i \in I \big\} \]
 and
 \begin{gather*}
  S' = S_{i_0, j_0, m_0} \cup
  \left( \bigcup_{\substack{j \in [k], \\ m \in [M], \\ (j,m) \ne (j_0,m_0)}} X_{i_0,j,m} \right) \cup
  \left( \bigcup_{i \in J \setminus \{ i_0 \}} Y_{i,1,1,1} \right) \cup
  \left( \bigcup_{\substack{i \in J \setminus \{ i_0 \}, \\ j \in [k] \setminus \{1\}, \\ m \in [M] \setminus \{1\}}} X_{i,j,m} \right) \\ \cup
  \left( \bigcup_{\substack{i \notin J, \\ j \in [k], \\ m \in [M]}} X_{i,j,m} \right) \cup
  W \text{.}
 \end{gather*}
 Then $S'$ is a cutset in $G_{t,k}-e$ with
  \[ |S'| = |S| + (nkM-1)aT + kMT' \]
 and similarly as before,
  \[ \omega \big( (G_{t, k}-e)-S' \big) > \frac{|S|}{t} + (nkM-1)bT + \alpha(G) + \left( \frac{MT'}{t}-1 \right)k = \frac{|S|}{t} + (nkM-1)bT + \frac{kMT'}{t} = \frac{|S'|}{t} \text{,} \]
 so $\tau(G_{t, k} - e) < t$.
 
 \bigskip
 
 \textit{Case 4:} $e$ connects two vertices of $V$.
 
 Since the case when $e$ is induced by $H^{i_0,j_0,m_0}$ for some $i_0 \in [n], j_0 \in [k],m_0 \in [M]$ was settled in Case~3, we can assume that there do not exist $i \in [n], j \in [k], m \in [M]$ for which $e$~is induced by $V_{i,j,m}$. By Lemma~\ref{alpha_crit_blowup}, the graph $G_{t, k}[V]$ is \mbox{$\alpha$-critical}, so in $G_{t, k}[V] - e$ there exists an independent vertex set $I$ of size $\alpha(G) + 1$. Let
  \[ J = \big\{ (i,j,l,m) \in [n] \times [k] \times [T] \times [M] \bigm| v_{i,j,l,m} \in I \big\} \text{,} \]
  \[ J'_1 = \big\{ (i,j,m) \in [n] \times [k] \times [M] \bigm| \exists! l \in [T]: ~ v_{i,j,l,m} \in I \big\} \text{,} \]
 and
  \[ J'_2 = \big\{ (i,j,m) \in [n] \times [k] \times [M] \bigm| \nexists l \in [T]: ~ v_{i,j,l,m} \in I \big\} \text{.} \]
 By the assumption that there do not exist $i \in [n], j \in [k], m \in [M]$ for which $e$~is induced by $V_{i,j,m}$,
  \[ J'_1 \cup J'_2 = [n] \times [k] \times [M] \text{,} \]
 so
  \[ S = \left( \bigcup_{(i,j,l,m) \in J} Y_{i,j,l,m} \right) \cup \left( \bigcup_{(i,j,m) \in J'_2} X_{i,j,m} \right) \cup W \]
 is a (well-defined) cutset in $G_{t,k}-e$.
 Then 
  \[ |S| = nkMaT + kMT' \]
 and similarly as before,
  \[ \omega \big( (G_{t, k}-e)-S \big) = nkMbT + \alpha(G) + 1 + \left( \frac{MT'}{t}-1 \right)k = \frac{nkMaT}{t} + \frac{kMT'}{t} + 1 > \frac{|S|}{t} \text{,} \]
 so $\tau(G_{t, k} - e) < t$.
 
 \bigskip
 
 Now let us assume that $G$ is not \mbox{$\alpha$-critical} with $\alpha(G) = k$, i.e.~either $\alpha(G) \ne k$ or even though $\alpha(G) = k$, the graph~$G$ is not \mbox{$\alpha$-critical}.
 
 \bigskip
 
 \textit{Case I:} $\alpha(G) > k$.
 
 Let $I$ be an independent vertex set of size $\alpha(G)$ in~$G$. Let
  \[ J = \big\{ i \in [n] \bigm| v_i \in I \big\} \]
 and
  \[ S = \left( \bigcup_{i \in J} Y_{i,1,1,1} \right) \cup \left(\bigcup_{\substack{i \in J, \\ j \in [k] \setminus \{1\}, \\ m \in [M] \setminus \{1\}}} X_{i,j,m} \right) \cup \left( \bigcup_{\substack{i \notin J, \\ j \in [k], \\ m \in [M]}} X_{i,j,m} \right) \cup W \text{.} \]
 Then $S$ is a cutset in $G_{t,k}-e$ with
  \[ |S| = nkMaT + kMT' \]
 and similarly as before,
  \[ \omega \big( (G_{t, k}-e)-S \big) = nkMbT + \alpha(G) + \left( \frac{MT'}{t}-1 \right)k > nkMbT + \frac{kMT'}{t} = \frac{nkMaT + kMT'}{t} = \frac{|S|}{t} \text{,} \]
 so $\tau(G_{t, k}) < t$, which means that $G_{t, k}$ is not minimally \mbox{$t$-tough}.
 
 \bigskip
 
 \textit{Case II:} $\alpha(G) \le k$.
 
 Since $G$ is not \mbox{$\alpha$-critical} with $\alpha(G) = k$ there exists an edge $e \in E(G)$ such that $\alpha(G - e) \le k$. By Lemma~\ref{dp_complete_main_lemma_t>=1}, the graph $(G-e)_{t, k}$ is \mbox{$t$-tough}, but it can be obtained from $G_{t, k}$ by edge-deletion, which means that $G_{t, k}$ is not minimally \mbox{$t$-tough}.
\end{proof}

Therefore the problem {\scshape Min-$1$-Tough} is DP-complete, so by Claim~\ref{almostmin1tough_equivalent_forms}, we can conclude the following.

\begin{corollary}
 Recognizing almost minimally \mbox{1-tough} graphs is DP-complete.
\end{corollary}

Let {\scshape Almost-Min-$1$-Tough} denote the problem of determining whether a given graph is almost minimally \mbox{1-tough}.
 
\section{Minimally $t$-tough graphs with $t \le 1/2$} \label{section_t<=1/2}

The case when $t \le 1/2$ is special in some sense: graphs with toughness at most~$1/2$ can have cut-vertices. Unlike in the previous cases, we reduce {\scshape Almost-Min-$1$-Tough} to this problem. But first, again, we construct an auxiliary graph.

\subsection{The auxiliary graph $H'_t$ when $t \le 1/2$}

Let $t \le 1/2$ be a positive rational number. Let $a, b$ be relatively prime positive integers such that $t = a/b$ and let $H_t$ be constructed as follows. Let 
 \[ V = \{ v_1, v_2, \ldots, v_a \} \text{,} \qquad U = \{ u_1, u_2, \ldots, u_{b-a} \} \text{,} \qquad W = \{ w_1, w_2, \ldots, w_a \} \text{.} \]
Place a clique on the vertices of~$V$, connect every vertex of~$V$ to every vertex of~$U$, and connect $v_i$ to~$w_i$ for all $i \in [n]$. See Figure~\ref{Fig:Ht_t<=1/2}.

\begin{figure}[H] 
\begin{center}
\begin{tikzpicture}
 \tikzstyle{vertex}=[draw,circle,fill=black,minimum size=4,inner sep=0]
 
 \draw (0,0) ellipse (0.4 and 1.5);
 \node at (0,-2) {$K_a$};
 
 \node[vertex] (v1) at (0,0.9) {};
 \node[vertex] (v2) at (0,0.3) {};
 \draw[fill] (0,-0.2) circle (0.3pt);
 \draw[fill] (0,-0.3) circle (0.3pt);
 \draw[fill] (0,-0.4) circle (0.3pt);
 \node[vertex] (vn) at (0,-0.9) {};
 
 \node at (0,2) {$V$};
 
 \node[vertex] (w1) at (-1,0.9) {};
 \node[vertex] (w2) at (-1,0.3) {};
 \draw[fill] (-1,-0.2) circle (0.3pt);
 \draw[fill] (-1,-0.3) circle (0.3pt);
 \draw[fill] (-1,-0.4) circle (0.3pt);
 \node[vertex] (wn) at (-1,-0.9) {};
 
 \node at (-1,2) {$W$};
 
 \draw[dashed] (2,0) ellipse (0.6 and 2);
 
 \node at (2,0) {$\overline{K}_{b-a}$};
 \node at (2.5,2.25) {$U$};
 
 \draw (v1) -- (w1);
 \draw (v2) -- (w2);
 \draw (vn) -- (wn);
 
 \draw (0.05,1.55) -- (1.9,2.05);
 \draw (0.05,-1.55) -- (1.9,-2.05);
 \draw (0.25,1.45) -- (1.6,-1.85);
 \draw (0.25,-1.45) -- (1.6,1.85);
\end{tikzpicture}
\caption{The graph~$H_t$ when $t \le 1/2$.}
\label{Fig:Ht_t<=1/2}
\end{center}
\end{figure}

\begin{prop}
 Let $t \le 1/2$ be a positive rational number. Then $\tau(H_t) = t$.
\end{prop}

\begin{proof} 
 Let $S$ be an arbitrary cutset of~$H_t$. We can assume that $S \cap (U \cup W) = \emptyset$ since removing any of the vertices of $U \cup W$ does not disconnect anything in the graph. Then $S \subseteq V$, so
  \[ \omega(H_t-S) = \begin{cases}
                      a+(b-a) = b & \text{if $S = V$,} \\
                      |S| + 1     & \text{if $S \ne V$,}
                     \end{cases} \]
 which implies that
 \[ \tau(H_t) = \min \left\{ \frac{|S|}{\omega(H_t - S)} ~ \middle| ~ S \subseteq V, S \ne \emptyset \right\} = \frac{a}{b} = t \text{.} \]
\end{proof}

By repeatedly deleting some edges of~$H_t$, eventually we obtain a minimally \mbox{$t$-tough} graph; let us denote it with~$H'_t$ (i.e. if there exists an edge whose deletion does not decrease the toughness, then we delete it). Obviously, we could not delete the edges between $V$ and~$W$, so the vertices of~$W$ still have degree~1 in~$H'_t$. 


Note that $V$ is a tough set of~$H'_t$. For further reference (to avoid confusion with other sets denoted by $V$), we introduce a new name for it.

\begin{notation} \label{notation_St}
 Let $S_t$ denote the tough set $V$ in $H'_t$.
\end{notation}

\subsection{``Gluing''}

\begin{defi}
 Let $H$ be a graph with a vertex~$u$ of degree~1, and let $v$ be the neighbor of~$u$. Let $G$ be an arbitrary graph, and ``glue'' $H - \{ u \}$ separately to all vertices of~$G$ by identifying each vertex of~$G$ with~$v$. Let $G \oplus_v H$ denote the obtained graph. (See Figure~\ref{Fig:gluing}.)
\end{defi}

\begin{figure}[H]
\begin{center}
\begin{tikzpicture}
 \tikzstyle{vertex}=[draw,circle,fill=black,minimum size=4,inner sep=0]
 
 \draw (0,0) ellipse (0.4 and 2);
 \node at (-0.5,2) {$G$};
 
 \node[vertex] (v1) at (0,1.35) {};
 \node[vertex] (v2) at (0,0.45) {};
 \draw[fill] (0,-0.35) circle (0.3pt);
 \draw[fill] (0,-0.45) circle (0.3pt);
 \draw[fill] (0,-0.55) circle (0.3pt);
 \node[vertex] (vn) at (0,-1.35) {};
 
 \node at (-0.65,1.35) {$v_1$};
 \node at (-0.65,0.45) {$v_2$};
 \node at (-0.65,-1.35) {$v_n$};
 
 \draw (0.7,1.35) ellipse (0.9 and 0.3);
 \draw (0.7,0.45) ellipse (0.9 and 0.3);
 \draw (0.7,-1.35) ellipse (0.9 and 0.3);
 
 \node at (2.5,1.35) {$H - \{ u \}$};
 \node at (2.5,0.45) {$H - \{ u \}$};
 \node at (2.5,-1.35) {$H - \{ u \}$};
\end{tikzpicture}
\caption{The graph $G \oplus_v H$.} \label{Fig:gluing}
\end{center}
\end{figure}

\subsection{The proof of Theorem~\ref{minttough_dp_complete} when $t \le 1/2$} \label{subsection_t<=1/2}

\begin{thm}
 For any positive rational number $t \le 1/2$, the problem {\scshape Min-\mbox{$t$-Tough}} is DP-complete.
\end{thm}

\begin{proof}
 Let $t \le 1/2$ be a positive rational number. In Proposition~\ref{min_t_tough_DP} we already proved that the problem {\scshape Min-\mbox{$t$-Tough}} is in DP. To show that it is DP-hard, we reduce {\scshape Almost-Min-$1$-Tough} to it.

 Let $G$ be an arbitrary graph and $n=|V(G)|$. Consider the graph~$H'_t$ and let $u \in U$ be an arbitrary vertex of~$H'_t$ having degree~1, and let $v$ be its neighbor. Let 
  \[ H''_t = H'_t - \{ u \} \]
 and let $H^i$ denote the $i$-th copy of~$H''_t$ ``glued'' to the vertex $v_i \in V(G)$ for all $i \in [n]$. (For examples see Figures~\ref{Fig:G1/2} and~\ref{Fig:G2/5}.)

 Now we show that $G$ is almost minimally \mbox{1-tough} if and only if $G_t=G \oplus_v H'_t$ is minimally \mbox{$t$-tough}. 

 First, let $G$ be almost minimally \mbox{1-tough}. We need to show that $G_t$~is minimally \mbox{$t$-tough}.

 Let $S \subseteq V(G_t)$ be an arbitrary cutset of~$G_t$. Let
  \[ C = C(S) = \big\{ i \in [n] \bigm| v_i \in V(G) \cap S \big\} \text{,} \]
  \[ c_i = \big| V(H^i) \cap S \big| - 1 \]
 for all $i \in C$, and 
  \[ d_i = \big| V(H^i) \cap S \big| \]
 for all $i \in [n] \setminus C$ (see~Figure~\ref{minttough_deterministic_constr_t<=1/2}). Finally, let
  \[ D = D(S) = \big\{ i \in [n] \setminus C \bigm| d_i > 0 \big\} \text{.} \]
 Using these notations it is clear that
 \[ |S| = |C| + \sum_{i \in C} c_i + \sum_{i \in D} d_i. \]
 By Proposition~\ref{obs_almostmin1tough}, the removal of $V(G) \cap S$ from~$G$ leaves at most $|V(G) \cap S| = |C|$ components. By Proposition~\ref{obs_below1}, the removal of $V(H'_t) \cap S$ from~$H'_t$ leaves at most $|V(H'_t) \cap S|/t$ components. But for all $i \in [n] \setminus C$ we have already counted the component of $G'_t-S$ which contains~$v_i$, and for all $i \in C$ we do not need to count the component~$\{ u \}$ of~$H'_t$. Hence
 \begin{gather*} 
  \omega(G_t-S) \le |C| + \sum_{i \in C} \left( \frac{c_i + 1}{t} - 1 \right) + \sum_{i \in D} \left( \frac{d_i}{t} - 1 \right) \\
  = \frac{|C| + \sum_{i \in C} c_i + \sum_{i \in D} d_i }{t} - |D| \le \frac{|S|}{t} \text{,}
 \end{gather*}
 which means that $\tau(G_t) \ge t$. 
 
 \begin{figure}[H]
 \begin{center}
 \begin{tikzpicture}[scale=1.5]
  \tikzstyle{vertex}=[draw,circle,fill=black,minimum size=4,inner sep=0]
  
  \draw (0,0) ellipse (0.6 and 3);
  \node at (-0.5,3) {$G$};
  
  \node[vertex] (v1) at (0,2.25) {};
  \node[vertex] (v2) at (0,1.35) {};
  \node[vertex] (v3) at (0,0.45) {};
  \node[vertex] (v4) at (0,-0.45) {};
  \draw[fill] (0,-1.25) circle (0.3pt);
  \draw[fill] (0,-1.35) circle (0.3pt);
  \draw[fill] (0,-1.45) circle (0.3pt);
  \node[vertex] (vn) at (0,-2.25) {};
  
  \node at (-0.85,2.25) {$v_1$};
  \node at (-0.85,1.35) {$v_2$};
  \node at (-0.85,0.45) {$v_3$};
  \node at (-0.85,-0.45) {$v_4$};
  \node at (-0.85,-2.25) {$v_n$};
  
  \draw (0.9,2.25) ellipse (1.1 and 0.3);
  \draw (0.9,1.35) ellipse (1.1 and 0.3);
  \draw (0.9,0.45) ellipse (1.1 and 0.3);
  \draw (0.9,-0.45) ellipse (1.1 and 0.3);
  \draw (0.9,-2.25) ellipse (1.1 and 0.3);
  
  \node at (2.5,2.25) {$H^1$};
  \node at (2.5,1.35) {$H^2$};
  \node at (2.5,0.45) {$H^3$};
  \node at (2.5,-0.45) {$H^4$};
  \node at (2.5,-2.25) {$H^n$};
  
  \draw (-0.15,2.25-0.15) rectangle (0.15,2.25+0.15);
  \draw (-0.15,-0.45-0.15) rectangle (0.15,-0.45+0.15);
  
  \draw (0.7,2.25-0.15) rectangle (1.5,2.25+0.15);
  \draw (0.7,1.35-0.15) rectangle (1.5,1.35+0.15);
  \draw (0.7,-0.45-0.15) rectangle (1.5,-0.45+0.15);
  \draw (0.7,-2.25-0.15) rectangle (1.5,-2.25+0.15);
  
  \node at (1.1,2.25) {$c_1$};
  \node at (1.1,1.35) {$d_2$};
  \node at (1.1,-0.45) {$c_4$};
  \node at (1.1,-2.25) {$d_n$};
 \end{tikzpicture}
 \caption{The graph~$G_t$ and the cutset $S$, when $t \le 1/2$.} \label{minttough_deterministic_constr_t<=1/2}
 \end{center}
 \end{figure}
 
 Now let $S_0$ be a tough set of~$H'_t$. Since $u$ has degree~1, we can assume that $u \notin S_0$. Let $S_0^1 \subseteq V(H^1)$ be the first copy of~$S_0$. Obviously, $S_0^1$ is a cutset in~$G_t$, and
  \[ \omega(G_t-S_0^1) = \omega(H'_t - S_0) = \frac{|S_0|}{t} = \frac{|S_0^1|}{t} \text{,} \]
 which means that $\tau(G_t) \le t$.
 
 Therefore, $\tau(G_t) = t$. 
 
 \bigskip
 
 Let $e \in E(G_t)$ be an arbitrary edge. We need to show that $\tau(G_t - e) < t$ for all $e \in E(G_t)$. Now we have two cases.
 
 \bigskip
 
 \textit{Case 1:} $e \in E(G)$.
 
 If $e$ is a bridge in~$G$, then $\tau(G_t - e) = 0 < t$. So assume that $e$ is not a bridge in~$G$. Let $S = S(e) \ne \emptyset$ be a vertex set in~$G$ guaranteed by Claim~\ref{almostmin1tough_equivalent_forms}, and for all $i \in [n]$ let $S_t^i \subseteq V(H^i)$ be the $i$-th copy of the tough set~$S_t$ defined in Notation~\ref{notation_St}. (Note that $v \in S_t$ and $u \notin S_t$.) Let 
  \[ J = J(S) = \big\{ i \in [n] \bigm| v_i \in S \big\} \]
 and consider the vertex set
  \[ S' = S \cup \left( \bigcup_{i \in J} S_t^i \right) = \bigcup_{i \in J} S_t^i \text{.} \]
 Then $S'$ is a cutset in $G_t-e$ with
  \[ |S'| = \sum_{i \in J} |S_t^i| = |S| \cdot |S_t| \]
 and
  \[ \omega \big( (G_t-e)-S' \big) > |S| + |S| \left( \frac{|S_t|}{t} - 1 \right) = \frac{|S| \cdot |S_t|}{t} = \frac{|S'|}{t} \text{,} \]
 which means that $\tau(G_t - e) < t$.
 
 \bigskip
 
 \textit{Case 2:} $e \in E(H^{i_0})$ for some $i_0 \in [n]$.
 
 If $e$ is a bridge in $H'_t$, then $\tau(G_t - e) = 0 < t$. So assume that $e$ is not a bridge in~$H'_t$ and let $S = S(e) \ne \emptyset$ be a vertex set in~$H'_t$ guaranteed by Proposition~\ref{minttoughlemma}. Again, since $u$ has degree~1, we can assume that $u \notin S$. Let $S^{i_0} \subseteq V(H^{i_0})$ be the $i_0$-th copy of~$S$. Obviously, $S^{i_0}$ is a cutset in $G_t-e$ and
  \[ \omega \big( (G_t - e) - S^{i_0} \big) = \omega \big( (H'_t - e) - S \big) > \frac{|S|}{t} = \frac{|S^{i_0}|}{t} \text{,} \]
 which means that $\tau(G_t - e) < t$.
 
 \bigskip
 
 Therefore, the graph~$G_t$ is minimally \mbox{$t$-tough}.

 \bigskip

 Now we show that if $G_t$ is minimally \mbox{$t$-tough}, then $G$ is almost minimally \mbox{1-tough}.

 First, we prove that $\tau(G) \ge 1$. Suppose to the contrary that $\tau(G) < 1$. Obviously, $G$ must be connected (otherwise $\tau(G_t) = 0 \ne t$), so there exists a cutset $S \subseteq V(G)$ in $G$ satisfying 
  \[\omega(G-S) > |S| \text{.} \]
 For all $i \in [n]$ let $S_t^i \subseteq V(H^i)$ be the $i$-th copy of the tough set~$S_t$ defined in Notation~\ref{notation_St}. (Note that $v \in S_t$ and $u \notin S_t$.) Let 
  \[ J = J(S) = \big\{ i \in [n] \bigm| v_i \in S \big\} \]
 and consider the vertex set
  \[ S' = S \cup \left( \bigcup_{i \in J} S_t^i \right) = \bigcup_{i \in J} S_t^i \text{.} \]
 Then $S'$ is a cutset in~$G_t$ with
  \[ |S'| = \sum_{i \in J} |S_t^i| = |S| \cdot |S_t| \]
 and
  \[ \omega(G_t-S') > |S| + |S| \left( \frac{|S_t|}{t} - 1 \right) = \frac{|S| \cdot |S_t|}{t} = \frac{|S'|}{t} \text{,} \]
 which means that $\tau(G_t) < t$ and that is a contradiction. So $\tau(G) \ge 1$.

 Now we prove that $\tau(G - e) < 1$ for all $e \in E(G)$. Let $e \in E(G)$ be an arbitrary edge. If $e$ is a bridge in~$G$, then $\tau(G - e) = 0 < 1$. Let us assume that $e$ is not a bridge in~$G$. Then $e$ is not a bridge in~$G_t$ either. Let $S = S(e) \ne \emptyset$ be a vertex set guaranteed by Proposition~\ref{minttoughlemma}. Consider the vertex set $S_0 = S \cap V(G)$. Since $e$ is a bridge in $G-S_0$ as well, $S_0$~is a cutset in $G-e$. Let
  \[ C = C(S)= \big\{ i \in [n] \bigm| v_i \in S_0 \big\} \text{,} \]
  \[ c_i = \big| V(H^i) \cap S \big| - 1 \]
 for all $i \in C$ and 
  \[ d_i = \big| V(H^i) \cap S \big| \]
 for all $i \in [n] \setminus C$. Let
  \[ D = D(S) = \big\{ i \in [n] \setminus C \bigm| d_i > 0 \big\} \text{.} \]
 Then
  \[ \omega \big( (G-e) - S_0 \big) > |S_0| = |C| \]
 must hold, otherwise, similarly as before,
 \begin{gather*}
  \omega \big( (G'-e) - S \big) \le |C| + \sum_{i \in C} \left( \frac{c_i + 1}{t} - 1 \right) + \sum_{i \in D} \left( \frac{d_i}{t} - 1 \right) \\
  = \frac{ |C| + \sum_{i \in C} c_i + \sum_{i \in D} d_i }{t} - |D| \le \frac{|S|}{t} \text{,}
 \end{gather*}
 which is a contradiction. So $\tau(G-e) < 1$.

 Therefore, $G$ is almost minimally \mbox{1-tough}.
\end{proof}

\section{Conclusion}

In this paper we proved that recognizing minimally \mbox{$t$-tough} graphs is DP-complete for any positive rational number~$t$. On the other hand, in~\cite{spec_graphclasses} we proved that in some special graph classes, this problem belongs to~P:
\begin{enumerate}
 \item[--] in the class of split graphs,
 \item[--] in the class of \mbox{claw-free} graphs when $t \le 1$, and 
 \item[--] in the class of \mbox{$2K_2$-free} graphs.
\end{enumerate}
These results are not really surprising since the toughness of split, or \mbox{claw-free}, or \mbox{$2K_2$-free} graphs can be computed in polynomial time, see~\cite{split_general}, \cite{clawfree}, and~\cite{2K2}, respectively.

It is also known that recognizing \mbox{$t$-tough} bipartite graphs is coNP-complete for any positive rational number $t \le 1$, see~\cite{recognize_toughness_bipartite} and~\cite{exact_toughness}, but determining the complexity of recognizing minimally \mbox{$t$-tough}, bipartite graphs is still open.

\section*{Acknowledgment}
The research of the first author was supported by National Research, Development and Innovation Office NKFIH, K-116769 and K-124171, by the National Research, Development and Innovation Fund (TUDFO/51757/2019-ITM, Thematic Excellence Program), and by the Higher
Education Excellence Program of the Ministry of Human Capacities in the frame of Artificial Intelligence research area of Budapest
University of Technology and Economics (BME FIKP-MI/SC).
The research of the second author was supported by National Research, Development and Innovation Office NKFIH, K-111827.
The research of the third author was supported by National Research, Development and Innovation Office NKFIH, K-124171.

\appendix
  
\section*{Appendix}

\begin{figure}[H]
 \begin{center}
 \begin{tikzpicture}
  \tikzstyle{vertex}=[draw,circle,fill=black,minimum size=3,inner sep=0]
  
  \node[vertex] (v1) at (0,1) [label=above:{$v_1$}] {};
  \node[vertex] (v2) at (0,0) [label=above right:{$v_2$}] {};
  \node[vertex] (v3) at (0,-1) [label=below:{$v_3$}] {};
  
  \node[vertex] (u1) at (1,1) [label=above:{$u_1$}] {};
  \node[vertex] (u2) at (1,0) [label=above:{$u_2$}] {};
  \node[vertex] (u3) at (1,-1) [label=below:{$u_3$}] {};
  
  \node[vertex] (w) at (2,0) [label=right:{$w$}] {};
  
  \draw[ultra thick] (v1) -- (v2);
  \draw[ultra thick] (v2) -- (v3);
  \draw[ultra thick] (v3) to [bend left=60] (v1);
  
  \draw (v1) -- (u1);
  \draw (v2) -- (u2);
  \draw (v3) -- (u3);
  
  \draw (u1) -- (w);
  \draw (u2) -- (w);
  \draw (u3) -- (w);
 \end{tikzpicture}
 \caption{The minimally 1-tough graph~$G'$ constructed in the beginning of Section~\ref{section_special_cases}, when $G \simeq K_3$. The edges of $K_3$ are drawn with thick lines.}
 \label{Fig:G'}
 \end{center}
\end{figure}

\begin{figure}[H]
 \begin{center}
 \begin{tikzpicture}
  \tikzstyle{vertex}=[draw,circle,fill=black,minimum size=3,inner sep=0]
  
  \node[vertex] (v111) at (0,4.25) [label=above right:{$v_{1,1,1}$}] {};
  \node[vertex] (v121) at (0,3.75) [label=below right:{$v_{1,2,1}$}] {};
  \node[vertex] (v211) at (0,2.25) [label=above right:{$v_{2,1,1}$}] {};
  \node[vertex] (v221) at (0,1.75) [label=below right:{$v_{2,2,1}$}] {};
  \node[vertex] (v311) at (0,0.25) [label=above right:{$v_{3,1,1}$}] {};
  \node[vertex] (v321) at (0,-0.25) [label=below right:{$v_{3,2,1}$}] {};
  \node[vertex] (v411) at (0,-1.75) [label=above right:{$v_{4,1,1}$}] {};
  \node[vertex] (v421) at (0,-2.25) [label=below right:{$v_{4,2,1}$}] {};
  \node[vertex] (v511) at (0,-3.75) [label=above right:{$v_{5,1,1}$}] {};
  \node[vertex] (v521) at (0,-4.25) [label=below right:{$v_{5,2,1}$}] {};
  
  \node[vertex] (u111) at (3,4.25) [label=above left:{$u_{1,1,1}$}] {};
  \node[vertex] (u121) at (3,3.75) [label=below left:{$u_{1,2,1}$}] {};
  \node[vertex] (u211) at (3,2.25) [label=above left:{$u_{2,1,1}$}] {};
  \node[vertex] (u221) at (3,1.75) [label=below left:{$u_{2,2,1}$}] {};
  \node[vertex] (u311) at (3,0.25) [label=above left:{$u_{3,1,1}$}] {};
  \node[vertex] (u321) at (3,-0.25) [label=below left:{$u_{3,2,1}$}] {};
  \node[vertex] (u411) at (3,-1.75) [label=above left:{$u_{4,1,1}$}] {};
  \node[vertex] (u421) at (3,-2.25) [label=below left:{$u_{4,2,1}$}] {};
  \node[vertex] (u511) at (3,-3.75) [label=above left:{$u_{5,1,1}$}] {};
  \node[vertex] (u521) at (3,-4.25) [label=below left:{$u_{5,2,1}$}] {};
  
  \node[vertex] (w11) at (7,0.25) [label=right:{$w_{1,1}$}] {};
  \node[vertex] (w21) at (7,-0.25) [label=right:{$w_{2,1}$}] {};
  
  \draw[ultra thick] (v111) -- (v121);
  \draw[ultra thick] (v211) -- (v221);
  \draw[ultra thick] (v311) -- (v321);
  \draw[ultra thick] (v411) -- (v421);
  \draw[ultra thick] (v511) -- (v521);
  
  \draw[ultra thick] (v111) to [bend right=60] (v211);
  \draw[ultra thick] (v111) to [bend right=60] (v221);
  \draw[ultra thick] (v121) to [bend right=60] (v211);
  \draw[ultra thick] (v121) to [bend right=60] (v221);
  
  \draw[ultra thick] (v211) to [bend right=60] (v311);
  \draw[ultra thick] (v211) to [bend right=60] (v321);
  \draw[ultra thick] (v221) to [bend right=60] (v311);
  \draw[ultra thick] (v221) to [bend right=60] (v321);
  
  \draw[ultra thick] (v311) to [bend right=60] (v411);
  \draw[ultra thick] (v311) to [bend right=60] (v421);
  \draw[ultra thick] (v321) to [bend right=60] (v411);
  \draw[ultra thick] (v321) to [bend right=60] (v421);
  
  \draw[ultra thick] (v411) to [bend right=60] (v511);
  \draw[ultra thick] (v411) to [bend right=60] (v521);
  \draw[ultra thick] (v421) to [bend right=60] (v511);
  \draw[ultra thick] (v421) to [bend right=60] (v521);
  
  \draw[ultra thick] (v511) to [bend left=60] (v111);
  \draw[ultra thick] (v511) to [bend left=60] (v121);
  \draw[ultra thick] (v521) to [bend left=60] (v111);
  \draw[ultra thick] (v521) to [bend left=60] (v121);

  \draw (v111) -- (u111);
  \draw (v121) -- (u121);
  \draw (v211) -- (u211);
  \draw (v221) -- (u221);
  \draw (v311) -- (u311);
  \draw (v321) -- (u321);
  \draw (v411) -- (u411);
  \draw (v421) -- (u421);
  \draw (v511) -- (u511);
  \draw (v521) -- (u521);
  
  \draw (u111) -- (w11);
  \draw (u211) -- (w11);
  \draw (u311) -- (w11);
  \draw (u411) -- (w11);
  \draw (u511) -- (w11);
  
  \draw (u121) -- (w21);
  \draw (u221) -- (w21);
  \draw (u321) -- (w21);
  \draw (u421) -- (w21);
  \draw (u521) -- (w21);
 \end{tikzpicture}
 \caption{The graph~$G'_{1,2}$ constructed in~Subsection~\ref{subsection_integer}, when $G \simeq C_5$. Since the graph $C_5$ is connected and $\alpha$-critical with $\alpha(C_5)=2$, the choice $k=2$ results in a minimally 1-tough graph. The edges of the ``blown-up'' $C_5$ are drawn with thick lines.}
 \label{Fig:G'12}
 \end{center}
\end{figure}

\begin{figure}[H]
 \begin{center}
 \begin{tikzpicture}
  \tikzstyle{vertex}=[draw,circle,fill=black,minimum size=3,inner sep=0]
  
  \node[vertex] (v111) at (0,2.25) [label=above right:{$v_{1,1,1}$}] {};
  \node[vertex] (v112) at (0,1.75) [label=below right:{$v_{1,1,2}$}] {};
  \node[vertex] (v211) at (0,0.25) [label=above right:{$v_{2,1,1}$}] {};
  \node[vertex] (v212) at (0,-0.25) [label=below right:{$v_{2,1,2}$}] {};
  \node[vertex] (v311) at (0,-1.75) [label=above right:{$v_{3,1,1}$}] {};
  \node[vertex] (v312) at (0,-2.25) [label=below right:{$v_{3,1,2}$}] {};
  
  \node[vertex] (u111) at (3,2.25) [label=above left:{$u_{1,1,1}$}] {};
  \node[vertex] (u112) at (3,1.75) [label=below left:{$u_{1,1,2}$}] {};
  \node[vertex] (u211) at (3,0.25) [label=above left:{$u_{2,1,1}$}] {};
  \node[vertex] (u212) at (3,-0.25) [label=below left:{$u_{2,1,2}$}] {};
  \node[vertex] (u311) at (3,-1.75) [label=above left:{$u_{3,1,1}$}] {};
  \node[vertex] (u312) at (3,-2.25) [label=below left:{$u_{3,1,2}$}] {};
  
  \node[vertex] (w11) at (5,0.25) [label=right:{$w_{1,1}$}] {};
  \node[vertex] (w12) at (5,-0.25) [label=right:{$w_{1,2}$}] {};
  
  \draw[ultra thick] (v111) -- (v112);
  \draw[ultra thick] (v211) -- (v212);
  \draw[ultra thick] (v311) -- (v312);
  
  \draw (u111) -- (u112);
  \draw (u211) -- (u212);
  \draw (u311) -- (u312);
  
  \draw[ultra thick] (v111) to [bend right=60] (v211);
  \draw[ultra thick] (v111) to [bend right=60] (v212);
  \draw[ultra thick] (v112) to [bend right=60] (v211);
  \draw[ultra thick] (v112) to [bend right=60] (v212);
  
  \draw[ultra thick] (v211) to [bend right=60] (v311);
  \draw[ultra thick] (v211) to [bend right=60] (v312);
  \draw[ultra thick] (v212) to [bend right=60] (v311);
  \draw[ultra thick] (v212) to [bend right=60] (v312);
  
  \draw[ultra thick] (v311) to [bend left=60] (v111);
  \draw[ultra thick] (v311) to [bend left=60] (v112);
  \draw[ultra thick] (v312) to [bend left=60] (v111);
  \draw[ultra thick] (v312) to [bend left=60] (v112);
  
  \draw (v111) -- (u111);
  \draw (v112) -- (u112);
  \draw (v211) -- (u211);
  \draw (v212) -- (u212);
  \draw (v311) -- (u311);
  \draw (v312) -- (u312);
  
  \draw (u111) -- (w11);
  \draw (u211) -- (w11);
  \draw (u311) -- (w11);
  
  \draw (u112) -- (w12);
  \draw (u212) -- (w12);
  \draw (u312) -- (w12);
 \end{tikzpicture}
 \caption{The graph~$G'_{2,1}$ constructed in~Subsection~\ref{subsection_integer}, when $G \simeq K_3$. Since the graph $K_3$ is 2-connected and $\alpha$-critical with $\alpha(K_3)=1$, the choice $k=1$ results in a minimally 2-tough graph. The edges of the ``blown-up'' $K_3$ are drawn with thick lines.}
 \label{Fig:G'21}
 \end{center}
\end{figure}

\begin{figure}[H]
 \begin{center}
 \begin{tikzpicture}
  \tikzstyle{vertex}=[draw,circle,fill=black,minimum size=3,inner sep=0]
  
  \node[vertex] (v1) at (90:0.5) [label=above right:{$v_1$}] {};
  \node[vertex] (v2) at (210:0.5) [label=above left:{$v_2$}] {};
  \node[vertex] (v3) at (330:0.5) [label=above right:{$v_3$}] {};
  
  \node[vertex] (u1) at (90:1.25) [label=above right:{$u_1$}] {};
  \node[vertex] (u2) at (210:1.25) [label=above left:{$u_2$}] {};
  \node[vertex] (u3) at (330:1.25) [label=above right:{$u_3$}] {};
  
  \draw[ultra thick] (v1) -- (v2) -- (v3) -- (v1);
  \draw (v1) -- (u1);
  \draw (v2) -- (u2);
  \draw (v3) -- (u3);
 \end{tikzpicture}
 \caption{The graph~$G_{1/2}$ constructed in~Subsection~\ref{subsection_1/b}, when $G \simeq K_3$. Since the graph $K_3$ is almost minimally 1-tough, this graph is minimally $1/2$-tough. The edges of $K_3$ are drawn with thick lines.}
 \label{Fig:G1/2}
 \end{center}
\end{figure}

\begin{figure}[H]
 \begin{center}
 \begin{tikzpicture}
  \tikzstyle{vertex}=[draw,circle,fill=black,minimum size=3,inner sep=0]
  
  \node[vertex] (v1) at (90:0.5) {};
  \node[vertex] (v12) at ($(v1)+(0:0.5)$) {};
  \node[vertex] (u11) at ($(v1)-(0:0.25)+(90:0.5)$) {};
  \node[vertex] (u12) at ($(v1)+(0:0.25)+(90:0.5)$) {};
  \node[vertex] (u13) at ($(v1)+(0:0.75)+(90:0.5)$) {};
  \node[vertex] (u14) at ($(v1)+(0:1.25)+(90:0.25)$) {};
  
  \node[vertex] (v2) at (210:0.5) {};
  \node[vertex] (v22) at ($(v2)+(120:0.5)$) {};
  \node[vertex] (u21) at ($(v2)-(120:0.25)+(210:0.5)$) {};
  \node[vertex] (u22) at ($(v2)+(120:0.25)+(210:0.5)$) {};
  \node[vertex] (u23) at ($(v2)+(120:0.75)+(210:0.5)$) {};
  \node[vertex] (u24) at ($(v2)+(120:1.25)+(210:0.25)$) {};
  
  \node[vertex] (v3) at (330:0.5) {};
  \node[vertex] (v32) at ($(v3)+(240:0.5)$) {};
  \node[vertex] (u31) at ($(v3)-(240:0.25)+(330:0.5)$) {};
  \node[vertex] (u32) at ($(v3)+(240:0.25)+(330:0.5)$) {};
  \node[vertex] (u33) at ($(v3)+(240:0.75)+(330:0.5)$) {};
  \node[vertex] (u34) at ($(v3)+(240:1.25)+(330:0.25)$) {};
  
  \draw[ultra thick] (v1) -- (v2) -- (v3) -- (v1);
  
  \draw (v1) -- (u11);
  \draw (v1) -- (u12);
  \draw (v1) -- (u13);
  \draw (v12) -- (u11);
  \draw (v12) -- (u12);
  \draw (v12) -- (u13);
  \draw (v12) -- (u14);
  
  \draw (v2) -- (u21);
  \draw (v2) -- (u22);
  \draw (v2) -- (u23);
  \draw (v22) -- (u21);
  \draw (v22) -- (u22);
  \draw (v22) -- (u23);
  \draw (v22) -- (u24);
  
  \draw (v3) -- (u31);
  \draw (v3) -- (u32);
  \draw (v3) -- (u33);
  \draw (v32) -- (u31);
  \draw (v32) -- (u32);
  \draw (v32) -- (u33);
  \draw (v32) -- (u34);
 \end{tikzpicture}
 \caption{The graph~$G_{2/5}$ constructed in~Subsection~\ref{subsection_t<=1/2}, when $G \simeq K_3$. Since the graph $K_3$ is almost minimally 1-tough, this graph is minimally $2/5$-tough. The edges of $K_3$ are drawn with thick lines.}
 \label{Fig:G2/5}
 \end{center}
\end{figure}

\begin{figure}[H]
 \begin{center}
 \begin{tikzpicture}
  \tikzstyle{vertex}=[draw,circle,fill=black,minimum size=3,inner sep=0]
  
  \node[vertex] (v111) at (0,3.5) {};
  \node[vertex] (v112) at (0,2.5) {};
  \node[vertex] (v211) at (0,0.5) {};
  \node[vertex] (v212) at (0,-0.5) {};
  \node[vertex] (v311) at (0,-2.5) {};
  \node[vertex] (v312) at (0,-3.5) {};
  
  \node[vertex] (a111) at (0.25,3.75) {};
  \node[vertex] (b111) at (0.25,3.25) {};
  \node[vertex] (c111) at (0.5,3.5) {};
  \node[vertex] (a112) at (0.25,2.75) {};
  \node[vertex] (b112) at (0.25,2.25) {};
  \node[vertex] (c112) at (0.5,2.5) {};
  \node[vertex] (a211) at (0.25,0.75) {};
  \node[vertex] (b211) at (0.25,0.25) {};
  \node[vertex] (c211) at (0.5,0.5) {};
  \node[vertex] (a212) at (0.25,-0.25) {};
  \node[vertex] (b212) at (0.25,-0.75) {};
  \node[vertex] (c212) at (0.5,-0.5) {};
  \node[vertex] (a311) at (0.25,-2.25) {};
  \node[vertex] (b311) at (0.25,-2.75) {};
  \node[vertex] (c311) at (0.5,-2.5) {};
  \node[vertex] (a312) at (0.25,-3.25) {};
  \node[vertex] (b312) at (0.25,-3.75) {};
  \node[vertex] (c312) at (0.5,-3.5) {};
  
  \node[vertex] (u111) at (3,3.5) {};
  \node[vertex] (u112) at (3,2.5) {};
  \node[vertex] (u211) at (3,0.5) {};
  \node[vertex] (u212) at (3,-0.5) {};
  \node[vertex] (u311) at (3,-2.5) {};
  \node[vertex] (u312) at (3,-3.5) {};
  
  \node[vertex] (w11) at (5,0.5) {};
  \node[vertex] (w12) at (5,-0.5) {};
  \node[vertex] (w'11) at (6,0.5) {};
  \node[vertex] (w'12) at (6,-0.5) {};
  
  \draw[ultra thick] (v111) -- (v112);
  \draw[ultra thick] (v211) -- (v212);
  \draw[ultra thick] (v311) -- (v312);
  
  \draw (u111) -- (u112);
  \draw (u211) -- (u212);
  \draw (u311) -- (u312);
  
  \draw[ultra thick] (v111) to [bend right=60] (v211);
  \draw[ultra thick] (v111) to [bend right=60] (v212);
  \draw[ultra thick] (v112) to [bend right=60] (v211);
  \draw[ultra thick] (v112) to [bend right=60] (v212);
  
  \draw[ultra thick] (v211) to [bend right=60] (v311);
  \draw[ultra thick] (v211) to [bend right=60] (v312);
  \draw[ultra thick] (v212) to [bend right=60] (v311);
  \draw[ultra thick] (v212) to [bend right=60] (v312);
  
  \draw[ultra thick] (v311) to [bend left=60] (v111);
  \draw[ultra thick] (v311) to [bend left=60] (v112);
  \draw[ultra thick] (v312) to [bend left=60] (v111);
  \draw[ultra thick] (v312) to [bend left=60] (v112);
  
  \draw (v111) -- (a111);
  \draw (v111) -- (b111);
  \draw (v112) -- (a112);
  \draw (v112) -- (b112);
  \draw (v211) -- (a211);
  \draw (v211) -- (b211);
  \draw (v212) -- (a212);
  \draw (v212) -- (b212);
  \draw (v311) -- (a311);
  \draw (v311) -- (b311);
  \draw (v312) -- (a312);
  \draw (v312) -- (b312);
  
  \draw (a111) -- (c111);
  \draw (b111) -- (c111);
  \draw (a112) -- (c112);
  \draw (b112) -- (c112);
  \draw (a211) -- (c211);
  \draw (b211) -- (c211);
  \draw (a212) -- (c212);
  \draw (b212) -- (c212);
  \draw (a311) -- (c311);
  \draw (b311) -- (c311);
  \draw (a312) -- (c312);
  \draw (b312) -- (c312);
  
  \draw (c111) -- (u111);
  \draw (c112) -- (u112);
  \draw (c211) -- (u211);
  \draw (c212) -- (u212);
  \draw (c311) -- (u311);
  \draw (c312) -- (u312);
  
  \draw (u111) -- (w11);
  \draw (u211) -- (w11);
  \draw (u311) -- (w11);
  
  \draw (u112) -- (w12);
  \draw (u212) -- (w12);
  \draw (u312) -- (w12);
  
  \draw (w11) -- (w'11);
  \draw (w11) -- (w'12);
  \draw (w12) -- (w'11);
  \draw (w12) -- (w'12);
 \end{tikzpicture}
 \caption{The graph~$G_{2/3,1}$ constructed in~Subsection~\ref{subsection_1/2<t<1}, when $G \simeq K_3$. Since the graph $K_3$ is 2-connected and $\alpha$-critical with $\alpha(K_3)=1$, the choice $k=1$ results in a minimally 2/3-tough graph. The edges of the ``blown-up'' $K_3$ are drawn with thick lines.}
 \label{Fig:G2/3,1}
 \end{center}
\end{figure}

\end{document}